\documentclass[11pt]{article}
\usepackage{amsthm}
\usepackage{amssymb}
\usepackage{amsmath}
\usepackage{todonotes}
\usepackage[margin=1in]{geometry}
\usepackage[utf8]{inputenc}
\usepackage{hyperref}

\usepackage[libertine]{newtxmath}
\usepackage{newtxtext}

\usepackage{pbox}
\usepackage{mathtools}

\theoremstyle{plain}
\newtheorem{thm}{Theorem}
\numberwithin{thm}{section}

\newtheorem{prop}[thm]{Proposition}
\newtheorem{claim}[thm]{Claim}
\newtheorem{lem}[thm]{Lemma}
\newtheorem{cor}[thm]{Corollary}
\newtheorem{conj}[thm]{Conjecture}
\newtheorem{ques}[thm]{Question}

\theoremstyle{definition}
\newtheorem{df}[thm]{Definition}

\theoremstyle{remark}
\newtheorem*{rem}{Remark}

\newcommand{\PCSP}{\operatorname{PCSP}}
\newcommand{\CSP}{\operatorname{CSP}}

\newcommand{\Pol}{\operatorname{Pol}}
\newcommand{\Ham}{\operatorname{Ham}}
\newcommand{\Hom}{\operatorname{Hom}}
\newcommand{\Zero}{\operatorname{Zero}}
\newcommand{\One}{\operatorname{One}}
\newcommand{\AND}{\operatorname{AND}}
\newcommand{\OR}{\operatorname{OR}}
\newcommand{\Par}{\operatorname{Par}}
\newcommand{\Maj}{\operatorname{Maj}}
\newcommand{\Malt}{\operatorname{AT}}
\newcommand{\setzero}{\operatorname{SET-ZERO}}
\newcommand{\setone}{\operatorname{SET-ONE}}
\newcommand{\EQUAL}{\operatorname{EQUAL}}
\newcommand{\Maltsev}{Alternating-Threshold}

\newcommand{\flip}{\operatorname{flip}}
\newcommand{\id}{\operatorname{id}}

\title{Promise Constraint Satisfaction: \\ Algebraic Structure and a Symmetric Boolean Dichotomy\footnote{Proceedings version appeared in SODA 2018~\cite{DBLP:conf/soda/BrakensiekG18}.}}
\author{Joshua Brakensiek\thanks{Department of Computer Science, Stanford University, Stanford, CA. Email: {\tt jbrakens@stanford.edu} This work was conducted while a student at Carnegie Mellon University.  Research supported in part by an REU supplement to NSF CCF-1526092.} \and Venkatesan Guruswami\thanks{Computer Science Department, Carnegie Mellon University, Pittsburgh, PA 15213. Email: {\tt guruswami@cmu.edu}. Research supported in part by NSF grant CCF-1526092 and CCF-1908125.} }
\date{}

\begin{document}

\maketitle
\thispagestyle{empty}
\begin{abstract}
  A classic result due to Schaefer (1978) classifies all constraint
  satisfaction problems (CSPs) over the Boolean domain as being either
  in $\mathsf{P}$ or $\mathsf{NP}$-hard. This paper considers a
  promise-problem variant of CSPs called PCSPs. A PCSP over a finite
  set of pairs of constraints $\Gamma$ consists of a pair
  $(\Psi_P, \Psi_Q)$ of CSPs with the same set of variables such that
  for every $(P, Q) \in \Gamma$, $P(x_{i_1}, \hdots, x_{i_k})$ is a
  clause of $\Psi_P$ if and only if $Q(x_{i_1}, \hdots, x_{i_k})$ is a
  clause of $\Psi_Q$. The promise problem
  $\operatorname{PCSP}(\Gamma)$ is to distinguish, given
  $(\Psi_P, \Psi_Q)$, between the cases $\Psi_P$ is satisfiable and
  $\Psi_Q$ is unsatisfiable. Many problems such as approximate graph
  and hypergraph coloring as well as the $(2+\epsilon)$-SAT problem
  due to Austrin, Guruswami, and H\aa
  stad~\cite{DBLP:journals/siamcomp/AustrinGH17} can be placed in this
  framework.

  \smallskip This paper is motivated by the pursuit of understanding
  the computational complexity of Boolean promise CSPs, determining
  for which $\Gamma$ the associated PCSP is polynomial-time tractable
  or $\mathsf{NP}$-hard. As our main result, we show that
  $\operatorname{PCSP}(\Gamma)$ exhibits a dichotomy (it is either
  polynomial-time tractable or $\mathsf{NP}$-hard) when the relations
  in $\Gamma$ are symmetric and allow for negations of variables. In
  particular, we show that every such polynomial-time tractable
  $\Gamma$ can be solved via either Gaussian elimination over $\mathbb
  F_2$ or a linear programming relaxation. We achieve our dichotomy
  theorem by extending the (weak) polymorphism framework of {Austrin, Guruswami, and H\aa stad} which
  itself is a generalization of the algebraic approach used by
  polymorphisms to study CSPs. In both the algorithm and hardness
  portions of our proof, we incorporate new ideas and techniques not
  utilized in the CSP case.
\end{abstract}

\newpage

\tableofcontents

\thispagestyle{empty}

\newpage

\setcounter{page}{1}
\section{Introduction}
\label{sec:intro}

A constraint satisfaction problem (CSP) over domain $D$ is specified
by a finite collection $\Lambda$ of relations over $D$, and is denoted
as $\CSP(\Lambda)$. An instance of $\CSP(\Lambda)$ consists of a set of
variables $V$ and a collection of constraints $\{(\tau, P)\}$ where $P
\in \Lambda$ and $\tau$ is a tuple of $k$ variables where $k$
is the arity of $P$ (i.e., $P \subseteq D^k$). The goal is find an
assignment $\sigma : V \to D$ that satisfies all constraints, i.e.,
$(\sigma(\tau_1),\dots,\sigma(\tau_k)) \in P$ for each constraint
$(\tau,P)$. In the optimization version, we seek an assignment that
maximizes the number of satisfied constraints.

Constraint satisfaction problems form a rich class of problems, and
have played a crucial role in the development of computational
complexity theory, starting from the NP-completeness of 3SAT to the
PCP theorem to the Unique Games Conjecture, all of which study the
intractability of a certain CSP.  Despite the large variety of
problems that can formulated as a CSP, it is remarkable that CSPs are
a class whose computational complexity one can dream of
understanding completely, for either the decision or optimization
version (including approximability in the latter case).  For Boolean
CSPs (those over domain $D=\{0,1\}$), Schaefer~\cite{Schaefer:1978} proved a
dichotomy theorem showing that such CSP is either polynomial time
solvable or NP-complete. Further, he gave a characterization of the
tractable cases -- a Boolean $\CSP(\Lambda)$ is in P in precisely six
cases, when every constraint in $\Lambda$ is (i) satisfied by all
$0$s, (ii) satisfied by all $1$s, (iii) a conjunction of 2CNF clauses,
(iv) a conjunction of Horn SAT clauses, (v) a conjunction of dual Horn
SAT clauses, and finally (vi) every constraint in $\Lambda$ is a
conjunction of affine constraints over $\mathbb{F}_2$. The Feder-Vardi
conjecture~\cite{feder-vardi} states that {such a complexity}
dichotomy holds for every $\CSP(\Lambda)$ over arbitrary finite
domains. Besides the Boolean domain, before 2017 it had been proved for a few
other cases, including CSPs over a domain of size
$3$~\cite{bulatov-ternary} and conservative CSPs (which contain all
unary relations)~\cite{bulatov-conservative,bulatov-conservative2}. The conjecture was recently resolved independently by Bulatov~\cite{Bulatov2017} and Zhuk~\cite{Zhuk2017}.

For the (exact) optimization version, a complete dichotomy theorem was
established in \cite{thapper-zivny} showing that for every collection
of relations $\Lambda$, the associated optimization problem is
tractable if and only if a certain basic linear programming relaxation
solves it, and it is NP-complete otherwise. The result in fact holds
for a generalization of Max CSP called \emph{valued} CSP, where each
constraint has a finite weight associated with it, and the goal is to
find a minimum value solution. When infinite weights are also allowed
(so some constraints have to be satisfied), it was shown that,
surprisingly, {the} dichotomy for ordinary CSPs {implies} a dichotomy
for this more general setting as well~\cite{KKR:2015}.  For
approximate optimization, a line of work exploring the consequences of
Khot's Unique Games Conjecture (UGC)~\cite{khot02} culminated in the
striking result~\cite{prasad} (see also \cite{BCR:2015}) that for
every CSP, there is a canonical semidefinite programming relaxation
which delivers the optimal worst-case approximation ratio, assuming
the UGC.

In this work we are interested in a potential complexity dichotomy for
\emph{promise} constraint satisfaction problems (PCSPs). A promise
{constraint satisfaction problem} $\PCSP(\Gamma)$ is specified
by a finite collection $\Gamma = \{ (P_i,Q_i) \}_i$ of \emph{pairs} of
relations with each $P_i \subseteq Q_i$. Let $\Lambda =\{ P_i\}_i$ and
$\Lambda' = \{Q_i\}_i$, Suppose we are given a satisfiable instance of
$\CSP(\Lambda)$ --- while finding a satisfying assignment might be
NP-hard, can we in polynomial time find a satisfying assignment when
the input is treated as an instance of $\CSP(\Lambda')$ (in the
obvious way, by replacing each constraint $P_i$ by the corresponding
$Q_i$)? The decision version $\PCSP(\Gamma)$ is the promise problem
where given an instance, we need to output Yes on instances that are
satisfiable as a $\CSP(\Lambda)$ instance, and output No on instances
that are unsatisfiable even as a $\CSP(\Lambda')$
instance.\footnote{PCSPs are not to be confused with a promise problem
  of a similar name studied by Ham and Jackson~\cite{Ham2016,Ham2017}
  which deal with a different kind of promise that certain partial
  assignments to a CSP are `extendable' to a full solution.} The
following challenge drives this work:

\begin{ques}\label{ques:PCSP-complexity-intro}
  For which $\Gamma$ is $\PCSP(\Gamma)$ polynomial-time tractable? For which $\Gamma$ is $\PCSP(\Gamma)$ $\mathsf{NP}$-hard? Must every $\Gamma$ fall into one of these two categories?
\end{ques}

The condition $P_i \subseteq Q_i$ guarantees that the satisfiability of a CSP with clauses
in the $P_i$'s implies the satisfiability of the CSP with the $P_i$'s
replaced by the corresponding $Q_i$'s. More generally one may consider two not
necessarily distinct domains $D_1$ and $D_2$, with $P_i \subseteq
D_1^k$ and $Q_i \subseteq D_2^k$ and a  map
$\sigma : D_1 \to D_2$ such that $\sigma(P_i) \subseteq Q_i$ (where $\sigma(P_i) =\{ (\sigma(a_1),\dots,\sigma(a_k)) \mid (a_1,\dots,a_k) \in P_i\}$. Note
that $\sigma$ is not necessarily injective. We omit the details of this more general presentation in this work (except briefly in the promise graph homomorphism example below).

To demonstrate the depth and far-reaching nature of the above question, we
provide some interesting examples which fall under this Promise CSP
framework. (Throughout the paper, we will use $[n] = \{1,
  \hdots, n\}$, $|x|$ to denote the Hamming weight (the number of 1s)
  of a Boolean vector $x \in \{0, 1\}$, and $e_i \in \{0, 1\}^n$ to
  denote the unique vector such that $(e_i)_j= 1$ if and only if $j =
  i$.)

\smallskip \noindent \textbf{\emph{CSPs.}}
Consider a PCSP $\Gamma = \{(P_i, Q_i) {: i \in [r]}\}$ such that $P_i = Q_i$ for all $i$. Then $\PCSP(\Gamma)$ is equivalent to the CSP decision problem $\CSP(\Gamma)$. Thus the above question in full generality subsumes the CSP dichotomy {theorem} as a special case.

\smallskip\noindent \textbf{\emph{Approximate graph coloring.}} Let $3 \le c \le t$ be positive integers, and consider the
relations $P = \{(a, b) \in [c]^2: a \neq b\}$ and $Q =
\{(a, b) \in [t]^2 : a \neq b\}$ with $D = \{1, \hdots,
t\}$.\footnote{Instead of having $P$ ignore a portion of the domain, we could present these more naturally in the homomorphism framework mentioned previously.} Then $(P,Q)$-PCSP is an instance of the approximate graph
coloring problem in which one needs to distinguish if the chromatic
number of a graph is at most $c$ or at least $t+1$.  The complexity of
$\PCSP(P,Q)$ is a notorious open problem; this problem is strongly believed to be
NP-hard for all $3 \le c \le t$, but the best NP-hardness in various
regimes~\cite{KLS,GK-sidma,Huang13,BrakensiekGuruswami:2016}
fall woefully short of establishing hardness for all $c$ and $t$,
especially when $c$ is small.

\smallskip\noindent\textbf{\emph{Hypergraph coloring.}} Generalizations of the coloring problem to the setting of
  hypergraphs also fall under this framework. The hardness of telling
  if a $3$-uniform hypergraph is $2$-colorable or not even
  $t$-colorable~\cite{DRS} (for any fixed $t$) is captured by
  $\PCSP(P,Q)$ where $P = \{1,2\}^3 \setminus \{(1,1,1),(2,2,2)\}$ and
  $Q = [t]^3 \setminus \{ (j,j,j) \mid j \in
  [t]\}$.

\smallskip\noindent \textbf{\emph{$(2+\epsilon)$-SAT.}} This problem
studied by \cite{DBLP:journals/siamcomp/AustrinGH17} corresponds to $\Gamma = \{(P_1, Q_1), (P_2, Q_2)\}$
where $(P_1, Q_1) = (\{x \in \{0, 1\}^{2k+1}, \ |x| \ge k\},
\{0,1\}^{2k+1}\setminus \{(0,\hdots, 0)\})$ (where $|x|$ is the
Hamming weight of $x$) and $(P_2, Q_2) =
(\{(0,1),(1,0)\},\{(0,1),(1,0)\})$. The purpose of $(P_2,Q_2)$ is so
that we can refer to some variables as negations of others. This
specific $\PCSP(\Gamma)$ was shown to be $\mathsf{NP}$-hard. On the
other hand, if we replace $P_1$ with $\{x \in \{0, 1\}^{2k+1}, |x| \ge
k + 1\}$, then $\PCSP(\Gamma)$ has a polynomial-time algorithm.

\smallskip\noindent\textbf{\emph{Hypergraph discrepancy.}}
Let $\Gamma = \{(P,Q)\}$ where $P = \{ x
\in \{0,1\}^{2k+1}, \ |x| \in \{k,k+1\}\}$ and $Q = \{0,1\}^{2k+1}
\setminus \{0^{2k+1},1^{2k+1}\}$. Then $\PCSP(\Gamma)$ was shown to be
hard in \cite{DBLP:journals/siamcomp/AustrinGH17}, which means that weak $2$-coloring of
hypergraphs with minimum discrepancy is hard. On the other hand, if the
arity is even and $P$ contains strings of equal number of $0$s and
$1$s, then $\PCSP(\Gamma)$ is tractable.

\smallskip\noindent\textbf{\emph{Promise graph homomorphism.}}
Let $G$ and $H$ be fixed directed graphs for which there is a
  homomorphism $\phi : G \to H$; that is, for all $(u, v) \in E(G)$,
  $(\phi(u), \phi(v)) \in E(H)$. Consider the following promise problem: given an input directed graph $X$, determine whether there exists a
  homomorphism $X \to G$ or there exists no homomorphism $X \to H$. We
  coin this question as the ``promise digraph homomorphism problem.''
This question can be encoded as a PCSP as follows. Let $P = E(G)$,
  the ordered pairs of vertices forming directed edges, and $Q =
  E(H)$. Any instance $\PCSP(P, Q)$ corresponds to the directed graph
  whose vertices are the variables and whose directed edges are the
  clauses.

  In the non-promise case {(}when $G = H$), the problem is known as the
  ``$H$-coloring" or the ``digraph homomorphism" problem (e.g.,
  \cite{HELL199092, feder-vardi}). Resolving the full CSP
  dichotomy {theorem} is equivalent to resolving the
  special case of digraph homomorphism
  problems~\cite{feder-vardi}. Special cases which {had} been resolved
  {before the full dichotomy} include the case $H$ is undirected~\cite{HELL199092} and when $H$
  has no sources or sinks~\cite{BartoKozikNiven2009}.

 Extending the polynomial-time equivalence of CSPs and digraph homomorphism problems,
  we show in Section~\ref{app:digraph} that every PCSP is polynomial-time equivalent
  to a promise digraph homomorphism problem. It follows that a
  dichotomy theorem for promise digraph homomorphism problems
  is equivalent to a dichotomy theorem for all PCSPs.
  
  Note that even the undirected case of promise digraph homomorphism problem is a substantial
  generalization of the approximate graph coloring problem, which corresponds to
  $G = K_c$ and $H = K_t$ being cliques. We conjecture the following very general hardness result for promise graph homomorphism:
    \begin{conj}
      Let $G, H$ be undirected non-bipartite graphs with a homomorphism from $G$ to $H$. Then the promise digraph homomorphism problem associated with $G$ and $H$ is NP-hard.\footnote{If
  either $G$ or $H$ is bipartite, the problem is easy via graph
  2-coloring.}
    \end{conj}
    Since originally posing this conjecture, a number of special cases have been solved, see Section~\ref{subsec:subsequent-work}.
    
 \bigskip
Given that PCSPs generalize CSPs and a dichotomy theorem for CSPs over
arbitrary domains is itself a long elusive challenge, in this work we focus on
Question~\ref{ques:PCSP-complexity-intro} for relations over the
Boolean domain. Even in this restricted setting, Boolean promise CSPs have {a richer} structure from that of Boolean CSPs (see Section~\ref{subsec:full-dichotomy}), rendering proving a generalization of Schaefer's dichotomy quite difficult, if it is even true.
In this work, we build the groundwork
for the complexity classification of promise PCPs, and prove a dichotomy
for the case of \emph{symmetric} Boolean promise CSPs allowing negations (Theorem~\ref{thm:intro-1} below). Negations can be enforced if $(P,Q) \in \Gamma$ where $P=Q=\{(0,1),(1,0)\}$; we say such a $\Gamma$ \emph{allows negations} or is \emph{folded}. 
A collection of
relation pairs $\Gamma = \{(P_i,Q_i) {: i \in [r]}\}$ is symmetric if each $P_i$
and $Q_i$ is a symmetric relation. A relation $P$ is symmetric if
$(a_1,a_2,\dots,a_l) \in P$ iff $(a_{\pi(1)},\dots,a_{\pi(l)}) \in P$
for every permutation $\pi \in S_l$. Note that a symmetric relation
$P \subseteq \{0,1\}^l$ can be specified by a set $S \subseteq
\{0,1,\dots,l\}$ such that $P = \{ x \in \{0,1\}^l \mid |x| \in S\}$.

\begin{thm}[Main]
\label{thm:intro-1}
Let $\Gamma$ be a symmetric collection of Boolean relation pairs that allows negations. Then
$\PCSP(\Gamma)$ is either in P or NP-hard.
\end{thm}
While the symmetry requirement is a significant restriction, it is a
natural subclass that still captures several fundamental problems, such as
$k$-SAT, Not-All-Equal-$k$-SAT, $t$-out-of-$k$-SAT, Hypergraph
Coloring, Bipartiteness, Discrepancy minimization, etc. In all these
cases, whether a constraint is satisfied only depends on the number of
variables set to $1$ (negations can be enforced via the symmetric
relation $\{(0,1),(1,0)\}$). Note that Horn SAT is an example of a CSP
that is \emph{not} symmetric.

We establish Theorem~\ref{thm:intro-1} via a characterization of all
the tractable cases, and showing that everything else is NP-hard. To
describe our results in greater detail, and to highlight the
challenges faced in extending Schaefer's theorem to the land of
promise CSPs, we now turn to the algebraic approach to study
$\CSP(\Lambda)$ via polymorphisms of the underlying relations.

Polymorphisms are operations that preserve membership in a
relation. Formally, $f: \{0,1\}^m \to \{0,1\}$ is a polymorphism of $P
\subseteq \{0,1\}^k$, denoted $f \in \Pol(P)$, if for all
$(a^{(i)}_1,\dots,a^{(i)}_k) \in P$, $i=1,2,\dots,m$,\\
$\bigl(f(a^{(1)}_1,a^{(2)}_1,\dots,a^{(m)}_1) , \cdots ,
f(a^{(1)}_k,a^{(2)}_k,\dots,a^{(m)}_k) \bigr) \in P$. For a collection
$\Lambda$ of relations, $\Pol(\Lambda) = \cap_{P \in \Lambda}
\Pol(P)$. Remarkably, the complexity of $\CSP(\Lambda)$ is completely
captured by $\Pol(\Lambda)$. The \emph{Galois
  correspondence}~\cite{jeavons88} states that if $\Pol(\Lambda')
\subseteq \Pol(\Lambda)$ then any relation in $\Lambda$ can be build from $\Lambda'$ using 
primitive positive reductions (see Section~\ref{app:Galois}), from which it follows that $\CSP(\Lambda)$ reduces to
$\CSP(\Lambda')$.  Note that all dictator functions (called
projections in the decision-CSP literature), $f(x_1,\dots,x_m) = x_j$ for some $j$,
always belong to $\Pol(\Lambda)$.

The algebraic dichotomy {theorem, proved by Bulatov~\cite{Bulatov2017} and Zhuk~\cite{Zhuk2017},} states that $\CSP(\Lambda)$ is in
$\mathsf{P}$ iff $\Pol(\Lambda)$ contains a certain kind of
polymorphism known as a weak-near unanimity operator
(see~\cite{BartoKozik:2009}). The NP-hardness half of this
{theorem was resolved a decade before the full dichotomy} (see~\cite{BartoKozik:2009} and
\cite{BJK05, Bulatov2000, Maroti2008}): if
$\Pol(\Lambda)$ fails to have a weak near-unanimity
operator,\footnote{According to~\cite{BartoKozik:2009}, a function $f : D^n \to D$, $n \ge 2$, is a weak near-unanimity
  operator if for all $a, b \in D$, $f(a, \hdots, a) = a$, a property
  known as idempotence, and $f(b, a, \hdots, a) = f(a, b, a, \hdots,
  a) = \cdots = f(a, \hdots, a, b).$ A
  simple example of a weak near-unanimity operator is the Boolean
  majority function on an odd number of variables.} then
there exists an NP-hard reduction.

The algebraic formulation of Schaefer's dichotomy theorem states that
a Boolean $\CSP(\Lambda)$ is tractable if $\Pol(\Lambda)$ contains
one of the six functions: constant $0$, constant $1$, Majority on $3$
variables, Boolean AND, Boolean OR, or parity of $3$
variables;\footnote{The first two cases are CSPs satisfied trivially
  by the all $0$s or all $1$s assignment; Majority corresponds to
  2SAT; AND and OR to Horn SAT and dual Horn SAT; and Parity to linear
  equations mod $2$.}  otherwise $\CSP(\Lambda)$ is NP-complete. We
refer the reader to the article by Chen~\cite{Chen:2009} for an
excellent contemporary treatment of Schaefer's theorem for Boolean
domains in the language of polymorphisms. For larger domains, there
has been a lot exciting recent progress, including the resolution of
the bounded width conjecture by Barto and
Kozik~\cite{BartoKozik:2009,BartoKozik:2014} proving a precise
characterization of when a natural local consistency algorithm works
for $\CSP(\Lambda)$ in terms of the structure of $\Pol(\Lambda)$.

Generalizing the situation for CSPs, it is natural to hope the
complexity of PCSPs {will also have some algebraic
  structure. Austrin, Guruswami and H\aa
  stad}~\cite{DBLP:journals/siamcomp/AustrinGH17} {pioneered
  the study of \emph{polymorphisms of PCSPs}}.\footnote{Austrin,
  Guruswami and H\aa stad referred to these objects as \emph{weak
    polymorphisms.}} A function $f : \{0,1\}^m \to \{0,1\}$ is
polymorphism for a pair of relations $(P,Q)$, denoted
$f \in \Pol(P,Q)$, if $f$ maps any $m$ inputs in $P$ to an output in
$Q$. When $P=Q$, this is just the notion of a CSP polymorphism for
$P$.

In Section~\ref{app:Galois}, we prove a generalization of the Galois
correspondence from CSPs to promise CSPs, establishing that the
complexity of a PCSP is captured by its polymorphisms (closely
related results, albeit without complexity-theoretic applications in
mind, are established by Pippenger~\cite{Pippenger2002}). Therefore
polymorphisms are the `right' approach to study the complexity of
promise CSPs.  When studying promise CSPs under the lens of polymorphisms, however, several challenges surface that didn't exist
in the world of CSPs. From an algebraic point of view, 
polymorphisms are no longer closed under composition (because after one
application, we no longer have an assignment in $P$, but rather for a
different relation $Q$). In universal algebra parlance, polymorphisms of PCSPs do not necessarily form a ``clone.''\footnote{Recently, the term \emph{clonoid} or \emph{minion} has been used to describe this large class of families.}  The dichotomy theorem for
Boolean CSPs can avail of a classification of all Boolean clones which
dates back to 1941~\cite{post} (again, see \cite{Chen:2009} for a
crisp presentation). In the world of promise CSPs, polymorphisms
belong to a broader class of {algebraic structures}, and it is a lot more
challenging to understand their structure; see
Section~\ref{subsec:full-dichotomy}.
 
From a complexity point of view, the distinction between easy and hard
is now more nuanced; the existence of any particular polymorphism doesn't
itself imply tractability. Indeed, for the $(2+\epsilon)$-SAT problem
mentioned earlier, majority of small arity is a polymorphism even
though the promise CSP is NP-hard.  At an intuitive level, we might
expect a PCSP to be easy if there are polymorphisms that
``genuinely'' depend on a lot of variables, and hard if a few
variables exert a lot of influence on the function. The precise way to
formalize this notion that captures the boundary between tractable and
hard is not yet clear. In \cite{DBLP:journals/siamcomp/AustrinGH17}, hardness was shown when
the only polymorphisms were juntas, functions which depends on a
bounded number of coordinates; in this work we relax this condition to
the existence of a small number of coordinates setting all of which to
$0$ fixes the function.

 In addition to establishing the hardness of many natural PCSPs, we
 also demonstrate the existence of new polynomial-time tractable
 PCSPs. As an example, consider a hypergraph $H = (V, E)$ such that
 all of its edges have bounded valence (but not all the valences need
 to be the same). Furthermore, for each $e = \{v_1, \hdots, v_k\} \in
 E$, we specify a hitting number $h_e \in \{1, \hdots, k-1\}$. Then,
 it is {polynomial-time} tractable to distinguish between the
 following two cases (1) there exists a two-coloring of the vertices
 of $H$ such that for all $e \in E$ the number of vertices of the
 {first color} is exactly $h_e$ and (2) every two-coloring of the
 vertices of $H$ leaves a monochromatic hyperedge. Formally, this is a
 PCSP with relations of the form $P = \{x \in \{0,1\}^k \mid |x| =
 a\}$ for any choice of $0 < a < k$ and $Q = \{x \in \{0,1\}^k\mid |x| \in
 \{1,2,\dots,k-1\}\}$. In essence, this PCSP is a hypergraph
 generalization of what makes 2-coloring for graphs efficient. The
 algorithm for solving this problem is based on linear
 programming. Unlike other CSPs and PCSPs, the proof of correctness
 uses the \emph{\Maltsev{} polymorphism}, a function which takes as
 input $x_1, \hdots, x_L \in \{0, 1\}$ ($L$ odd) and returns whether
 $x_1 - x_2 + \cdots - x_{L-1} + x_L$ is positive. In the Boolean
 setting for CSPs or PCSPs, this is the first non-symmetric
 polymorphism known to yield a polynomial time algorithm.\footnote{If
   we remove the symmetric condition on the relations, it turns our
   that non-symmetric polymorphisms are the norm, even in the Boolean
   case, see Section~\ref{app:theory}.}

 We now informally state the main dichotomy (for a formal statement see Theorem \ref{thm:main-result}) in two ways. First, we give an explicit characterization in terms of the structure of the PCSP itself. For simplicity, we only state a subset of the main result in this form. {In particular, we allow for the ``negation of variables,'' that if $x_i$ can be used in a clause, then $\bar{x}_i$ can also be used and for the ``setting of constants,'' we can force a variable in a clause to be the constant $0$ or $1$.}

 \begin{thm}\label{res:simple}
   Let $P \subseteq Q \subset \{0, 1\}^k$ be a symmetric pair of
   relations. Let $\Gamma$ contain the promise relation $(P, Q)$ as
   well as allow for negation of variables--$(\{(0, 1), (1, 0)\},
   \{(0, 1), (1, 0)\})$--and the setting of constants--$(\{0\},
   \{0\})$ and $(\{1\}, \{1\})$--(e.g. $x_i = 0$). Let $S = \{|x| \mid x \in P\}$ and $T  = \{|x| \mid x \in Q\}$. Furthermore, assume that $S \cap \{1, \hdots, k - 1\}$ is nonempty. Then, $\PCSP(\Gamma)$ is polynomial-time tractable if

   a) $S \subseteq \{\ell \in [k] \mid \ell \text{ odd}\} \subseteq T$
   or $S \subseteq \{\ell \in \{0\}\cup[k] \mid \ell \text{ even}\}
   \subseteq T$ or

   b) $T \supseteq \{0, 1, \hdots, k\} \cap \{2\min S - k + 1, \hdots,
   2\max S - 1\}$ or

   c) $|S| = 1$ and $T \supseteq \{1, \hdots, k - 1\}$.
   
   Otherwise, $\PCSP(\Gamma)$ is $\mathsf{NP}$-hard.
 \end{thm}

 Second, we give {a complete and elegant} ormulation of the dichotomy in terms of the polymorphisms of the PCSP instead of the PCSP itself.

 \begin{thm}[{Theorem~\ref{thm:main-result}}]\label{res:poly}
   Let $\Gamma$ be a family of pairs of symmetric relations which allows for negations as well as the setting of constants. Then, $\PCSP(\Gamma)$ is polynomial-time tractable if

   a) The Parity of $L$ variables {or the negation of Parity of $L$ variables} is a polymorphism of $\Gamma$
   for all odd $L$ or

   b) The Majority of $L$ variables {or the negation of Majority of $L$ variables} is a polymorphism of $\Gamma$
   for all odd $L$ or

   c) The \Maltsev{} of $L$ variables {or the negation of \Maltsev{} of $L$ variables} is a polymorphism of $\Gamma$ for all odd $L$.

   Otherwise, $\PCSP(\Gamma)$ is $\mathsf{NP}$-hard.
 \end{thm}

 {Note that we also need to consider the negations of polymorphisms in the full result.} To illustrate this, consider $P = \{x \in \{0, 1\}^5 \mid |x| \in \{2\}\}$ and $Q = \{x \in \{0, 1\}^5 \mid |x| \in \{1, 2, 3, 5\}\}$ and allow negations. {It is not hard to show that none of Parity, Majority, or \Maltsev{} on $11$ variables is in $\Pol(P, Q)$,} yet $\PCSP(\Gamma)$ is polynomial-time {tractable.} {Instead the negation of Parity or `anti-Parity' on $L$ variables is a polymorphism of $\Pol(P, Q)$ for all odd $L$.} In Section \ref{subsec:idempotence}, we show how to handle this technical {issue.}

\subsection{Proof Overview}\label{subsec:proof-overview}

The proof of the main theorem consists of three major parts. First, in Section \ref{sec:alg} we show that any PCSP which has one of these families of functions as a polymorphism--Parity, Majority, or Alternating-Threshold, or their {negations, which we call \emph{non-idempotent polymorphisms}}--has a polynomial time algorithm. The algorithms we demonstrate are quite general in that the only assumption we make is the existence of polymorphisms, in particular we do \emph{not} rely on the symmetry assumption. For Parity, we show that the problem can be reduced to an ordinary CSP with Parity as a polymorphism, and thus Schaefer's theorem can be invoked. For Majority and Alternating-Threshold, such a tactic cannot be used. Instead, we show how these problems can be written as linear programming relaxations. Surprisingly, identical algorithms are used in both cases to solve the decision problem. They do diverge, however, if one desires to use the LP relaxation to also find a solution when the PCSP is satisfiable. To deal with {non-idempotent polymorphisms}, we show in Section \ref{subsec:idempotence} how these PCSPs with {non-idemotent polymorphisms} can be reduced in polynomial time to PCSPs with their negations (that is, the {``normal'' or idempotent} polymorphisms), which we already know are polynomial{-}time solvable.

Second, in Section \ref{sec:classification}, for every symmetric PCSP
with negations that does not have the entirety of any of the mentioned
families of polymorphisms, we show that its polymorphisms
are `lopsided.' More precisely, we show that there exists a constant
$C$, only dependent on the type of the PCSP, such that for all
polymorphisms of the PCSP, there are $C$ coordinates such that setting all $C$ of those coordinates to the same value fixes the value of the polymorphism. We say that such polymorphisms are ``$C$-fixing.'' The general philosophy of the argument is as follows. First, since our $\Gamma$ fails to have \Maltsev{} on $L$ variables for some odd $L$ as a polymorphism, there is some $(P, Q) \in \Gamma$ responsible for this exclusion. Using a nuanced combinatorial argument, we attempt to classify the polymorphisms of $(P, Q)$ given that $P$ and $Q$ are symmetric. To simplify the proof, we first show that we may transform $(P, Q)$ into a canonical $(P', Q')$ without losing any polymorphisms (see Lemma \ref{lem:no-maltsev}). From this, we show that all polymorphisms $f$ of $\Gamma$ have the property that either $f(e_i)$ differs from $f(0, \hdots, 0)$ for a bounded number of $e_i$ or a substantial portion of $f$ is structured like the Parity polymorphism.  Since we assume that Parity of $L'$ variables is not a polymorphism of $\Gamma$ for some odd $L'$, we can show that the latter situation is impossible. Using another $(P'', Q'') \in \Gamma$ which fails to have Majority as a polymorphism, and after simplifying $(P'', Q'')$ to a canonical form, we can use arguments inspired from \cite{DBLP:journals/siamcomp/AustrinGH17} to obtain additional structural information which yields that all polymorphisms are $C$-fixing. We crucially exploit that $(P, Q)$ and $(P'', Q'')$ are symmetric to get these structural properties, but do not assume anything about the other clauses of $\Gamma$.

Finally, since we have pinned down the nature of the polymorphisms in these believed-to-be-hard PCSPs, in Section \ref{sec:hardness}, we prove the $\mathsf{NP}$-hardness of these PCSPs. We prove this by reducing from Label Cover, a well-known problem to reduce from for hardness of approximation proofs. This part of the proof is based on an argument of \cite{DBLP:journals/siamcomp/AustrinGH17}, but we greatly simplify how projection constraints are handled. With this hardness result established, the main theorem is proved.

\subsection{Must there be a Dichotomy?}\label{subsec:full-dichotomy}

Extending this dichotomy from the symmetric case to the full Boolean
case presents significant challenges, some of which perhaps suggest
that a dichotomy does not exist. {Compared to when this
  manuscript was first written, the possibility of a dichotomy seems
  more likely, but it has become more clear what the hurdles will be
  establishing such a result.}

{One challenge is the large variety of families of polymorphisms to consider.} In Section~\ref{app:theory}, we provide necessary and sufficient
conditions, adapting a result of Pippenger~\cite{Pippenger2002}, for a
family of functions $\mathcal F$ to satisfy $\mathcal F =
\Pol(\Gamma)$ for some PCSP $\Gamma$ (not necessarily Boolean). These
conditions, known as \emph{projection-closure} and
\emph{finitization}, are extremely flexible, allowing for an extremely
rich variety of polymorphisms. Note that these results liberate
us from ever thinking about $\Gamma$, and instead we can think
entirely in terms of establishing the easiness/hardness of
projection-closed, finitized families of functions. There is, however,
a caveat: there is a huge amount of freedom in finitizable,
projection-closed families of functions!

{However, a fair counterargument is that the polymorphisms themselves are too fine-grained of a perspective even for the CSP dichotomy. Rather, the \emph{identities} which the polymorphisms satisfy (e.g., whether the polymorphism is symmetric in its coordinates, etc.) are more important (e.g.,~\cite{Barto2017}). In particular, the recent work of Barto, Bulin, Krokhin, and Opršal~\cite{BartoBulinKrokhinEtAl2019} has shown it suffices to classify PCSPs by their ``height-1 identities.''}

{That said,} even in the Boolean setting{, both the algorithmic and hardness sides of a dichotomy resolution seem more daunting.}  For {an algorithmic} example, let $p$ be any prime number and let $S \subset \{0, \hdots, p - 1\}$ be a non-empty strict subset. Then, for each $L \in \mathbb N$ define $f^{(L)} : \{0, 1\}^L \to \{0, 1\}$ such that $f^{(L)}(x) = 0$ if  $|x| \mod p \in S$, and  $f^{(L)}(x) = 1$ if  $|x| \mod p \not\in S$. {Therefore,} if a $\PCSP$ $\Gamma$ has $f^{(L)} \in \Pol(\Gamma)$ for infinitely many $L$, then $\PCSP(\Gamma)$ can be efficiently solved using Gaussian elimination over $\mathbb F_p$! {Thus, algorithms for Boolean $\PCSP$ seem to require algorithms for CSPs over arbitrarily large domains.} {The subsequent work of \cite{BartoBulinKrokhinEtAl2019} shows that in fact CSP algorithms over infinite domains can be necessary! This alone, combined with the fact that dichotomy theorems for infinite-domain CSPs do not exist in general (e.g., \cite{bodirsky2012complexity}), should give some pause.} 

{The hardness side of the dichotomy is also vastly more complex for PCSPs than for CSPs. The hardness side of the CSP dichotomy theorem was resolved relatively early (e.g.,~\cite{BJK05}) with a gadget reduction from 3-SAT. On the other hand, even for Boolean PCSPs, the use of a PCP-like theorem for hardness seems necessary. For instance, the hardness of $(2+\epsilon)$-SAT in \cite{DBLP:journals/siamcomp/AustrinGH17} relied on the PCP theorem.\footnote{{Recent unpublished work of Barto-Kozik shows a ``Baby PCP theorem'' which might suffice in some cases.}}
Some other hardness results for PCSPs are based on very strong, conjectural versions of the PCP theorem. For instance, strong hardness results for approximate graph coloring rely on variants of the Unique Game Conjecture~\cite{DMR}, although new paths to circumventing this assumption have recently appeared~\cite{BartoBulinKrokhinEtAl2019,KrokhinOprsal2019,WrochnaZivny2019}. A recent work showed that Boolean ordered PCSPs, whose polymorphisms must be monotone functions, exhibit a dichotomy~\cite{BGS21}, but the hardness side relied on the ``Rich $2$-to-$1$ conjecture'' due to~\cite{BKM21}, which is a significant strenghtening of the $2$-to-$1$ conjecture of Khot~\cite{khot02}.}

{This leads the authors to conjecture that if a dichotomy of PCSPs does hold, then the algorithmic side will be resolved (long) before the hardness side. In fact, follow-up work~\cite{BGWZ20} has given a polynomial time (decision) algorithm for any PCSP with infinitely many symmetric polymorphisms; and it is not currently known if there is a tractable Boolean PCSP which is \emph{not} solved by this algorithm.\footnote{{Jakub Opršal showed this algorithm does not work for general non-Boolean (P)CSPs (see remark in \cite{BGWZ20}); however he conjectures that a natural extension of the algorithm may work for all CSPs~\cite{jakub}.}}}

\subsection{Subsequent Work}\label{subsec:subsequent-work}
 Since the original version of this paper~\cite{DBLP:conf/soda/BrakensiekG18}, there have been numerous follow-up works. Most notably, \cite{BulinKrokhinOprsal2019} developed a universal algebraic theory of PCSPs. In particular, generalizing the Galois correspondence of this paper, they show that the \emph{identities} which the polymorphisms satisfy suffice to determine the complexity. As a result, they were able to show that 3 versus 5 approximate graph coloring is NP-hard. A subsequent revision~\cite{BartoBulinKrokhinEtAl2019} contains additional results, including a result that the polynomial-time complexity of 1-in-3-SAT versus NAE-3-SAT cannot be attributed to any finite-domain CSP.

In this work, we assume that negations of variables are allowed in our dichotomy, a recent work~\cite{FicakKozikOlsakEtAl2019} removes this assumption by relaxing the notion of $C$-fixing.

On the complexity of $\PCSP(G, H)$, where $G \to H$ are non-bipartite graphs, \cite{KrokhinOprsal2019} showed that this problem is NP-hard whenever $H = K_3$. Interestingly, there hardness proof uses topologically inspired techniques which have been further developed in the recent work {of} \cite{WrochnaZivny2019}. See also~\cite{krokhin2020topology}.

On the algorithmic side, follow-up works of the authors \cite{BrakensiekGuruswami2019,BrakensiekGuruswami2020} have fleshed out many of the algorithmic techniques applied in this paper. In particular, \cite{BrakensiekGuruswami2020} show that any PCSP with infinitely many symmetric polymorphisms has an efficient decision algorithm, although constructing an efficient search algorithm in general is still an open question. The exact class of PCSPs for which this algorithm solves was classified in an expanded version~\cite{BGWZ20}.

\subsection{Organization}
 In Section~\ref{sec:PCSP}, we formally define the notion
of a PCSP as well as other tools and terminology which we will need in
investigating PCSPs. In Section~\ref{sec:alg}, we prove the
algorithmic portion of the main theorem. In
Section~\ref{sec:classification}, we give a structural characterization of PCSPs that is used to show hardness. In Section~\ref{sec:hardness}, we use the results of
Section~\ref{sec:classification} to complete the
$\mathsf{NP}$-hardness results of the main theorem. In Section~\ref{app:theory}, we prove some more general facts about PCSPs, including that polymorphisms precisely capture the computational complexity of PCSPs as well as classify the possible families of polymorphisms of PCSPs.

\section{Promise Constraint Satisfaction Problems}\label{sec:PCSP}

We develop a theory of the complexity of promise constraint satisfaction problems (PCSPs) analogous to that of `ordinary' CSPs such as found in \cite{Chen:2009}. We need to formally define what we mean by a PCSP.

\begin{df}\label{df:promise-relation}
Let $D$ be a finite domain. A \textit{relation} of arity $k$ is a subset $P \subseteq D^k$. A promise relation is a pair of relations $(P, Q)$ of arity $k$ such that $P \subseteq Q$.
\end{df}
We say that a relation is \textit{Boolean} if $D = \{0, 1\}$ (or more generally $|D| = 2$). For a given relation $P$, we will refer to it both as a subset of $D^k$ as well as its indicator function $P : D^k \to \{0,1\}$ ($P(x) = 1$ iff $x \in P$). It should be clear from context which notation for $P$ we are using. If $(P, Q)$ is a promise relation then $P(x) = 1 \implies Q(x) = 1$. When $P = Q$, the promise relation $(P, Q)$ is analogous to the relation $P$ in a CSP. In fact, when it is clear that we are referring to promise relations, we let $P$ denoted the promise relation $(P, P)$.

\begin{df}\label{df:(P,Q)-PCSP}
  Let $(P, Q)\subseteq D^k \times D^k$ be a promise relation. A \textit{$(P,Q)$-PCSP} is a pair of formulae $(\Psi_P, \Psi_Q)$, each with $m$ clauses on the variables $x_1, \hdots, x_n$ along with a variable-choice function $\ell : [m] \times [k] \to [n]$, such that $\Psi_P(x_1, \hdots, x_n) = \bigwedge_{i=1}^{m} P(x_{\ell(i, 1)}, x_{\ell(i, 2)}, \hdots, x_{\ell(i, k)})$ and $\Psi_Q(x_1, \hdots, x_n) = \bigwedge_{i=1}^{m} Q(x_{\ell(i, 1)}, x_{\ell(i, 2)}, \hdots, x_{\ell(i, k)})$. {Further,} let $\Gamma = \{(P_i, Q_i) \subseteq D^{k_i} \times D^{k_i} : i \in [r]\}$ be a set of promise relations over $D$ of possibly distinct arities. For $i \in [r]$ let $(\Psi_{P_i}, \Psi_{Q_i})$ be a $(P_i,Q_i)$-PCSP so that each PCSP is on the same variable set $x_1, \hdots, x_n$. A \textit{$\Gamma$-PCSP} is then a pair of formula $(\Psi_P, \Psi_Q)$ such that $\Psi_P(x_1, \hdots, x_n) = \bigwedge_{i=1}^r \Psi_{P_i}(x_1, \hdots, x_n)$ and $\Psi_Q(x_1, \hdots, x_n) = \bigwedge_{i=1}^r \Psi_{Q_i}(x_1, \hdots, x_n)$.
\end{df}

We say that $(\Psi_P, \Psi_Q)$ is \textit{satisfiable} if there exists $(x_1, \hdots, x_n) \in D^n$ such that $\Psi_P(x_1, \hdots, x_n) = 1$. That is, $\Psi_P$ is satisfiable in the usual sense. We say that $(\Psi_P, \Psi_Q)$ is \textit{unsatisfiable} if $\Psi_Q$ is unsatisfiable, for all $(x_1, \hdots, x_n) \in D^n$, $\Psi_Q(x_1, \hdots, x_n) = 0$. Since the clauses involve promise relations, any satisfying assignment to $\Psi_P$ is necessarily a satisfying assignment to $\Psi_Q$, so no $\Gamma$-PCSP can be simultaneously satisfiable and unsatisfiable. Despite that, it is possible for the PCSP to be neither satisfiable nor unsatisfiable. As an extreme case, consider $P = \{\}$ and $Q = D^k$ then every $(P, Q)$-PCSP $(\Psi_P, \Psi_Q)$ has the property $\Psi_P$ is unsatisfiable but $\Psi_Q$ is satisfiable, so the PCSP is neither satisfiable or unsatisfiable. As such, the main computational problem we seek to study is a promise decision problem.
\begin{df}\label{df:PCSP(Gamma)}
  Let $\Gamma = \{(P_i, Q_i) \subseteq D^{k_i} \times D^{k_i}\}$ be a set of promise relations. $\PCSP(\Gamma)$ is the following promise decision problem. Given a $\Gamma$-PCSP $\Psi = (\Psi_P, \Psi_Q)$, output YES if $\Psi$ is satisfiable and output NO if $\Psi$ is unsatisfiable.
\end{df}
Note that $\PCSP(\Gamma)$ is in $\mathsf{promiseNP}$ since we can easily check in polynomial time if an assignment satisfies $\Psi_P$. We implicitly allow repetition of the variables in a specific clause. We show in Section~\ref{app:repetition} that removing this assumption does not meaningfully change the complexity of the problem.

\begin{rem}
  An equivalent notation which is used in subsequent works (e.g., \cite{BartoBulinKrokhinEtAl2019}) is to consider sets $\mathbb{A} = \{P_i : i \in [r]\}$ and $\mathbb B = \{Q_i : i \in [r]\}$ and denote $\PCSP(\Gamma)$ by $\PCSP(\mathbb A, \mathbb B)$. We keep with the former notation in this paper.
\end{rem}

\subsection{Polymorphisms}

As it can be quite cumbersome to find a direct $\mathsf{NP}$-hardness reduction for $\PCSP(\Gamma)$, we study the combinatorial properties of a set of functions known as \textit{polymorphisms}, which have served well as a proxy for the computational complexity of PCSPs \cite{DBLP:journals/siamcomp/AustrinGH17, BrakensiekGuruswami:2016}.

\begin{df}\label{df:weak-poly}
  Let $(P, Q) \in D^k \times D^k$ be a promise relation. A polymorphism of $(P, Q)$ is a function $f : D^L \to D$ such that for all $(x_1^{(1)}, \hdots, x_k^{(1)}), \hdots, (x_1^{(L)}, \hdots, x_k^{(L)}) \in P$ then $(f(x_1^{(1)}, \hdots, x_1^{(L)}), \hdots, f(x_k^{(1)}, \hdots, x_k^{(L)})) \in Q$. Denote this set of polymorphisms as $\Pol(P, Q)$. If $\Gamma = \{(P_i, Q_i) {: i \in [r]}\}$ is a set of promise relations, then $f : D^L \to D$ is a polymorphism of $\Gamma$ iff $f$ is a polymorphisms of $(P_i, Q_i)$ for all $i$.
\end{df}

We let $\Pol(\Gamma)$ denote the set of polymorphisms of $\Gamma$. Note that the projection maps $\pi_i(x) = x_i$ are polymorphisms of every promise relation. Further note that $\Pol(\Gamma) = \bigcap_{(P_i, Q_i) \in \Gamma} \Pol(P_i, Q_i).$

When $P_i = Q_i$, these polymorphisms are the polymorphisms studied in the CSP literature (e.g. \cite{Chen:2009}). Sadly, when $P_i \neq Q_i$, the polymorphisms are no longer easily composable, so we no longer have necessarily that our polymorphisms form a clone. {However,} we still have one key property of a clone, that the polymorphisms are closed under {\emph{projections}\footnote{Also called {\emph{minors}} in the literature.}}. 

\begin{df}\label{df:proj}
  Let $f : D^{{R}} \to D$ be a polymorphism of a family of promise
  relations $\Gamma$. Let $\pi : [{R}] \to [L]$ be a map. A
  \textit{projection} $f^\pi : D^{{L}} \to D$ is the map
  $(f^ \pi)(x) = f(y), \forall i, y_i = x_{\pi(i)}.$
  It is straightforward to verify that $f^ \pi \in \Pol(\Gamma)$.
\end{df}

For the remainder of the article, we assume that $D = \{0, 1\}$.

\begin{df}\label{df:folding}
  Let $f : \{0, 1\}^L \to \{0, 1\}$ be a polymorphism of a family of Boolean promise relations $\Gamma$. We say that $f$ is \textit{folded} if $f(x) = \neg f(\bar{x})$ for all $x \in \{0, 1\}^L$. We say that a family of promise relations $\Gamma$ is folded if all of its polymorphisms are folded. 
\end{df}

It is straightforward to show that if $\Gamma$ contains the NOT relation ($P = Q = \{(0, 1), (1, 0)\}$) then all polymorphisms are folded. Furthermore, note that projections of folded functions are also folded.

We will also view a polymorphism as \textit{generating} a set of promise relations $\Gamma'$ from a set of relations $\Gamma$.
\begin{df}\label{df:generating}
  Let $f : \{0, 1\}^L \to \{0, 1\}$ be a polymorphism, and let $P \subseteq \{0, 1\}^k$ be a relation. Define $f(P)$ to be
  \[
  f(P) := \{x \in \{0, 1\}^k : \text{exist $x^{(1)}, \hdots, x^{(L)} \in P$ such that $x_i = f(x^{(1)}_i, \hdots, x^{(L)}_i)$ for all $i \in [k] = \{1, \hdots, k\}$}\}
  \]
\end{df}

This could also be written as $f(P) := f(P^L)$. We often state that $x = f(x^{(1)}, \hdots, x^{(L)})$, where $x^{(i)} \in P$, as a shorthand for $x_i = f(x^{(1)}_i, \hdots, x^{(L)}_i)$ for all $i \in [k]$. Note that $f \in \Pol(P, Q)$ if and only if ${f}(P) \subseteq Q$.

What is the motivation for studying these polymorphisms? Roughly,
if $\Gamma$ has an interesting family of polymorphisms, then we
expect for that family to `beget' a polynomial-time algorithm for
$\PCSP(\Gamma)$. The following are examples of families of polymorphisms which will yield algorithms. For all of these functions,
we have that our domain is $x \in \{0, 1\}^L$.
\begin{itemize}
\setlength\itemsep{0em}
\item The zero and one functions: $\Zero_L(x) = 0$, $\One_L(x) = 1$.
\item The AND and OR functions: $\AND_L(x) = \bigwedge_{i=1}^L x_i$, $\OR_L(x) = \bigvee_{i=1}^L x_i$.
\item The Parity function: $\Par_L(x) = \bigoplus_{i=1}^L x_i$. ($L$ odd)
\item The Majority function: $\Maj_L(x) = {1}$ {if} $\sum_{i=1}^L x_i > L/2$ {and $0$ otherwise} ($L$ odd).
\item The \Maltsev{} function: $\Malt_L(x) = {1}$ {if} $\sum_{i=1}^L (-1)^{i-1} x_i > 0$ {and $0$ otherwise} ($L$ odd).
\end{itemize}
Note that except for the \Maltsev{} operator, all of these
polymorphisms appear in the modern treatment of Schaefer's
Theorem. Although the \Maltsev{} operator is a polymorphism of some
traditional Boolean CSPs, such as 2-coloring, in those cases it is
possible to show that Majority is also present as a polymorphism. We
will see later that this is not the case for PCSPs. Note that the
arity-3 \Maltsev{} operator would be considered a \emph{Mal'tsev
  operator} in traditional CSPs (e.g. \cite{BulatovDalmau2006}).

In addition to these polymorphisms, we also use the prefix `anti-' to refer to the negations of these functions. The `anti-' polymorphism will be denoted with a horizontal bar. For example, anti-parity is $\overline{\Par}_L(x) = \neg \Par_L(x)$. Note that the One function is the `anti-Zero' function and vice-versa. These polymorphisms appear due to technicalities of the nature of promise-CSPs. In Section \ref{subsec:idempotence}, we show that these anti-s can be transformed into normal polymorphisms.

In Section \ref{sec:alg}, we show that if $\Pol(\Gamma)$ contains any one of these infinite families of polymorphisms, then $\PCSP(\Gamma)$ is tractable.

\subsection{Decoding}
As mentioned in the introduction, one formulation of the Algebraic
{CSP} Dichotomy {Theorem} is that for any finite set of finite
(traditional) relations $\Gamma$, the decision problem on the
satisfiability of CSPs with clauses from $\Gamma$ is in $\mathsf P$ if
and only if $\Gamma$ a weak near-unanimity operator. In the case of
Promise CSPs, the picture is known to be not as clean. For example,
\cite{DBLP:journals/siamcomp/AustrinGH17, BrakensiekGuruswami:2016} both study \textsf{NP}-hard PCSPs
in which some of the polymorphisms depend non-trivially on
multiple coordinates. In both of those works, the polymorphisms
depend on a \textit{bounded} number of coordinates, either literally
or after correcting some noise. By utilizing these polymorphisms
as gadgets in a suitable probabilistically checkable proof, such as
Label Cover, hardness was obtained.

In \cite{BrakensiekGuruswami:2016}, we approached understanding these polymorphisms of NP-hard PCSPs using a \textit{robust decoding} framework which identified influential coordinates in these polymorphism in a manner amenable to Label Cover. In this paper, to identify influential coordinates we will use the concept of a \textit{$C$-fixing junta}.

\begin{df}
  Let $f : \{0,1\}^L \to \{0, 1\}$ be a folded polymorphism. We
  say that a folded polymorphism is {a \textit{$C$-fixing junta} (or just \emph{$C$-fixing})} if
  there exists $S \subseteq \{1, \hdots, L\}$ with $|S| \le C$ such
  that if $x \in \{0, 1\}^L$ satisfies $x_i = 0$ for all $i \in S$,
  then $f(x) = f(0, \hdots 0)$.
\end{df}

In Section \ref{sec:hardness}, we show that for any folded family of promise relations $\Gamma$ all of whose polymorphisms are $C$-fixing, then $\PCSP(\Gamma)$ is $\mathsf{NP}$-hard. In Section \ref{sec:classification}, we show for a large class of $\Gamma$ that their polymorphisms are $C$-fixing via combinatorial arguments.

\subsection{Idempotence} \label{subsec:idempotence}

Define a function $f : \{0, 1\}^L \to \{0, 1\}$ to be \textit{idempotent} if $f(0, \hdots, 0) = 0$ and $f(1, \hdots, 1) = 1$. We say that a family of promise relations $\Gamma$ is idempotent if all polymorphisms are idempotent.

\begin{prop}\label{prop:idempotence}
  For any relation $P \subseteq \{0,1\}^k$ and any idempotent function $f$, we have that $P \subseteq f(P)$.
\end{prop}
\begin{proof}
  For every $x \in P$ note that $x_i = f(x_i, \hdots, x_i)$; thus $x \in f(P)$.
\end{proof}
We say that $f$ \textit{generates} the promise relation $(P, f(P))$ from $P$. If $\Gamma = \{P_i {: i \in [r]}\}$ is a set of relations, then $f(\Gamma) = \{(P_i, f(P_i)) {: i \in [r]}\}.$ 
Essentially by definition, $f(\Gamma)$ has $f$ as a polymorphism.

We can force the polymorphisms of a family of promise relations $\Gamma$ to be idempotent by adding in the unary promise relations $\setzero = (\{(0)\}, \{(0)\})$ and $\setone = (\{(1)\}, \{(1)\})$.

\begin{prop}\label{prop:more-idempotence}
  For any family of promise relations $\Gamma$, the set of idempotent promise relations of $\Gamma$ is exactly $\Pol(\Gamma \cup \{\setzero, \setone\})$.
\end{prop}
\begin{proof}
  For any idempotent $f \in \Pol(\Gamma)$, we have $f(0, \hdots, 0) = 0$ and $f(1, \hdots, 1) = 1$, so $f \in \Pol(\setzero)$ and $f \in \Pol(\setone)$. Likewise, every polymorphism of $\Gamma \cup \{\setzero, \setone\}$ is idempotent.
\end{proof}

For a relation $Q$, define $\neg Q = \{\bar{x} : x \in Q\}.$ If $(P, Q)$ is a promise relation, it is not longer clear that $(P, \neg Q)$ is a promise relation, because we might not have that $P \subseteq \neg Q$. If we assume \textit{non-degeneracy} and that $(P, Q)$ has a non-idempotent polymorphism, then this is the case.

\begin{df}
A function $f : \{0, 1\}^L \to \{0, 1\}$ is \textit{non-degenerate} if $f(0,\hdots, 0) \neq f(1, \hdots, 1)$. A family of promise relations $\Gamma$ is non-degenerate if all of its polymorphisms are \textit{non-degenerate}.
\end{df}

One can verify that $\Gamma$ is non-degenerate if and only if $\Zero_1, \One_1 \not\in \Pol(\Gamma)$. 

\begin{prop}\label{prop:even-more-idempotence}
Let $(P, Q)$ be a promise relation with a non-degenerate, non-idempotent polymorphism $f : \{0, 1\}^L \to \{0, 1\}$. Then, $(P, \neg Q)$ is a promise relation, and $\neg f$ is a idempotent polymorphism of this promise relation.
\end{prop}
\begin{proof}
Since $f$ is non-idempotent and non-degenerate, we have that $f(0,\hdots, 0) = 1$ and $f(1, \hdots, 1) = 0$. Thus, for any $x \in \{0, 1\}^L$, we have that $f(x, \hdots, x) = \bar{x}$. Since $f \in \Pol(P, Q)$, we thus have that $\neg P \subseteq Q$. Thus, $P \subseteq \neg Q$, so $(P, \neg Q)$ is a promise relation. It is easy to then see that for any $x^1, \hdots, x^L \in P$, since $f(x^1, \hdots, x^L) \in Q$, $\neg f(x^1, \hdots, x^L) \in \neg Q$. Thus, $\neg f$, which is idempotent, is a polymorphism of $(P, \neg Q)$.
\end{proof}

Thus, if a non-degenerate family of promise relations $\Gamma$ has at least one non-idempotent polymorphism, we may define $\neg \Gamma = ((P_i, \neg Q_i) : (P_i, Q_i) \in \Gamma)$ as another family of promise relations. Note that since $\Gamma$ always has idempotent polymorphisms (the {dictators/}projections), $\neg \Gamma$ thus has non-idempotent polymorphisms, so $\neg(\neg \Gamma)$ exists and is equal to $\Gamma$. Thus, the idempotent polymorphisms of $\Gamma$ are exactly the non-idempotent polymorphisms of $\neg \Gamma$ and vice-versa. We can formally show that the idempotent polymorphisms of $\Gamma$ and $\neg \Gamma$ capture the computational complexity of $\PCSP(\Gamma)$.

\begin{lem}\label{lem:idempotent}
  Let $\Gamma$ be a non-degenerate family of promise relations with at least one non-idempotent polymorphism. Let $\Gamma' = \Gamma \cup \{\setzero, \setone\}$ and $\Gamma'' = (\neg \Gamma) \cup \{\setzero, \setone\}$. Then
  \begin{enumerate}
  \item $\Pol(\Gamma) = \Pol(\Gamma') \cup (\neg \Pol(\Gamma''))$, where $\neg \Pol(\Delta) = \{\neg f : f \in \Pol(\Delta)\}$.
  \item If $\PCSP(\Gamma')$ or $\PCSP(\Gamma'')$ is polynomial-time tractable, then so is $\PCSP(\Gamma)$.
  \end{enumerate}
\end{lem}

\begin{proof}
  1. In Proposition \ref{prop:more-idempotence}, we have that the polymorphisms of $\Pol(\Gamma')$ are exactly the idempotent polymorphisms of $\Pol(\Gamma)$. From Proposition \ref{prop:even-more-idempotence} and the subsequent discussion, we have that the non-idempotent polymorphisms of $\Gamma$ are exactly the negations of the idempotent polymorphisms of $\neg \Gamma$ which are the polymorphisms of $\Gamma''$. Thus, $\Pol(\Gamma) = \Pol(\Gamma') \cup (\neg \Pol(\Gamma''))$ since every polymorphism of $\Gamma$ is either idempotent or non-idempotent and every polymorphism of $\Gamma'$ and $\Gamma''$ is idempotent.

  2. Since $\Gamma \subseteq \Gamma'$, we have that if $\PCSP(\Gamma')$ is polynomial-time tractable, then $\PCSP(\Gamma)$ is polynomial-tractable by applying the exact same algorithm. Now, assume that $\PCSP(\Gamma'')$ is polynomial-time tractable. Since  ${\neg \Gamma \subseteq \Gamma''}$, we have that $\PCSP(\neg \Gamma)$ is polynomial-time tractable. Consider an instance $\Psi = (\Psi_P, \Psi_Q)$ of $\PCSP(\Gamma)$. Let $\Psi^{\neg} = (\Psi_P, \Psi_{\neg Q})$ be an instance of ${\PCSP(\neg\Gamma)}$ in which every $Q_i$ clause of $\Psi_Q$ is replaced with a $\neg Q_i$ clause. Clearly $\Psi_P$ is satisfiable if and only if $\Psi_P$ is satisfiable and $\Psi_Q$ is satisfiable if and only if $\Psi_{\neg Q}$ is satisfiable (a satisfying assignment to one is the negation of a satisfying assignment to the other). Thus, {$\Psi$ is satisfiable if and only if $\Psi^{\neg}$ is satisfiable}. Thus, if we run the algorithm for  $\PCSP(\neg \Gamma)$ which decides $\Psi^{\neg}$, we have also solved the problem in polynomial time for $\Psi$.
\end{proof}

In the proceeding sections, we utilize this lemma repeatedly so that we do not need to separately consider the non-idempotent polymorphisms.

\subsection{Symmetric PCSPs}
The primary focus of this paper is the study of $\Gamma$ in which every relation is \textit{symmetric}.

\begin{df}
A relation $P \subseteq \{0, 1\}^k$ is symmetric if for all $x \in P$ and all permutations $\sigma : \{1, \hdots, k\} \to \{1, \hdots, k\}$, we have that $(x_{\sigma(1)}, \hdots, x_{\sigma(k)}) \in P$. We say that a family of promise relations $\Gamma = \{(P_i, Q_i) {: i \in [r]}\}$ is symmetric if $P_i$ and $Q_i$ are symmetric for all $i$.
\end{df}

For a symmetric family of promise relations $\Gamma = \{(P_i, Q_i){: i \in [r]}\}$,
we have that each $P_i$ and $Q_i$ is uniquely determined by its arity
and the Hamming weights of the elements. We let $\Ham_{k}(S) = \{x \in
\{0, 1\}^k : |x| \in S\}$ denote these sets. For example, $\text{NOT}
= \{(0, 1), (1, 0)\} = \Ham_2(\{1\})$. Furthermore, the idempotence
relations $\text{SET-ZERO}$ and $\text{SET-ONE}$ are also symmetric,
so adding these relations to a symmetric family of promise relations
preserves that the family is symmetric. The following property of
symmetric relations helps us when working with polymorphisms.

\begin{prop}\label{prop:sym}
  Let $P$ be a symmetric relation. Let $f : \{0, 1\}^L \to \{0, 1\}$ be any function. Then, $f(P)$ is symmetric.
\end{prop}
\begin{proof}
For any $y \in f(P)$ and permutation $\sigma : \{1, \hdots, k\} \to \{1, \hdots, k\}$, consider the $x^1, \hdots, x^L \in P$ such that $f(x^1, \hdots, x^L) = y$. If we apply $\sigma$ to the coordinates of $x^1, \hdots, x^L$, they will stay in $P$ (since $P$ is symmetric). Furthermore, $f$ applies to these permuted variables with be $\sigma$ applied to the coordinates of $y$.
\end{proof}

In the remainder of the paper, we prove the following result. Note that Theorem \ref{res:poly} follows as a corollary.

\begin{thm}[Main Result]\label{thm:main-result}
  Let $\Gamma$ be a folded, symmetric, finite family of promise relations. If at least one of $\Par_L$, $\Maj_L$, $\Malt_L$, $\overline{\Par}_L$, $\overline{\Maj}_L,$ or  $\overline{\Malt}_L$ is a polymorphism of $\Gamma$ for all odd $L$, then $\PCSP(\Gamma)$ is polynomial-time tractable. Otherwise, $\PCSP(\Gamma)$ is $\mathsf{NP}$-hard.
\end{thm}

\section{Efficient Algorithms}\label{sec:alg}
In this section, we show that if a finite collection of promise relations $\Gamma$ has a polymorphism of a certain kind, then there exists a polynomial-time algorithm for solving $\PCSP(\Gamma)$. Note that we need not assume that the relations of $\Gamma$ are symmetric. We let $k$ refer to the maximum arity of any relation of $\Gamma$. 

\subsection{Zero, One, AND, OR, Parity}
In each of these cases, we will reduce the PCSP $\Gamma$ to a traditional CSP $\Gamma'$ with the same polymorphism, which we can then solve in polynomial time by virtue of Schaefer's theorem. See \cite{Chen:2009}.

\begin{lem}
  Let $\Gamma = \{(P_i, Q_i) : i \in \{1, \hdots, \ell\}\}$ be a finite family of promise relations, each of arity at most $k$. Suppose that $\Gamma$ has $f$ as a polymorphism, in which $f \in \{\Zero_1, \One_1, \AND_{2^k}, \OR_{2^k}, \Par_{2^k+1}\}$. Then, $\PCSP(\Gamma)$ is polynomial-time tractable.
\end{lem}

\begin{proof}
  If for some $(P_i, Q_i) \in \Gamma$, $P_i$ is the empty relation, we can check if our $\Gamma$-PCSP has a $P_i$ clause and reject, otherwise, we run the polynomial time algorithm for the promise relation family $\Gamma \setminus \{(P_i, Q_i) {: i \in [r]}\}$. Thus, we may without loss of generality assume that no $P_i$ of $\Gamma$ is the empty relation.

  For each possible $f$, we reduce the family of promise relations
  $\Gamma$ to $\Gamma' = \{R_i = f(P_i)\cup P_i : i \in \{1, \hdots,
  \ell\}\}$. We must have that $P_i \subseteq R_i \subseteq Q_i$, so
  the reduction is immediate (replace each $(P_i, Q_i)$ clause with a
  corresponding $R_i$ clause). We now show that $\CSP(\Gamma')$ is tractable in each case.
  
  \textbf{Case 1,} $f = \Zero_1$. For all $i \in \{1, \hdots, \ell\}$, note that $f(P_i) = \{(0, \hdots, 0)\}$. Thus, for all $R_i \in \Gamma'$, $R_i$ is closed under $\Zero_1$. Thus $\Gamma'$ has $\Zero_1$ as a polymorphism and so $\CSP(\Gamma')$ is polynomial-time tractable. {S}etting every variable to $0$ satisfies the $\Gamma'$-CSP.
  
  \textbf{Case 2,} $f = \One_1$. This is identical to Case 1, except $f(P_i) = \{(1, \hdots, 1)\}$.

  \textbf{Case 3,} $f = \AND_{2^k}$. Since $2^k \ge |P_i|$, the bitwise-AND of every subset of $P_i$ must be in $R_i$. Thus, we have that $R_i$ must be closed under the $\AND_2$ operator. Thus, $\Gamma'$ has $\AND_2$ as a polymorphism and is polynomial-time tractable.

  \textbf{Case 4,} $f = \OR_{2^k}$. Essentially the same as Case 3.

  \textbf{Case 5,} $f = \Par_{2^k+1}$. Since $2^k + 1 > |P_i|$, the bitwise-XOR of every odd-sized subset of $P_i$ is in $R_i$. Thus, $R_i$ is closed under the $\Par_3$ operator (the symmetric difference of 3 odd-sized subsets is an odd-sized subset). Thus, $\Gamma'$ has $\Par_3$ as a polymorphism and so it is polynomial-time tractable via a Gaussian-elimination algorithm. 
\end{proof}

\subsection{Majority and \Maltsev{}}

The algorithms in the previous section used the fact that
$\PCSP(\Gamma)$ has a tractable CSP $\Gamma'$ that is `sandwiched' by
$\Gamma$. If $\Gamma$ has the $\Maj_L$ or $\Malt_L$ polymorphism for
all odd $L$, it is no longer always the case that the polymorphisms of
$\Gamma$ contain the polymorphisms of a tractable CSP.\footnote{This
  follows from the fact that the composition of $\Maj_L$ (or
  $\Malt_L$) with itself is not a higher-arity $\Maj$ (or $\Malt$)
  function. In other words, the closure of the family of majority (or
  alternating-threshold) functions under identification do not form a
  clone.} Instead, we demonstrate tractability
by writing any $\Gamma$-PCSP $\Psi = (\Psi_P, \Psi_Q)$ as a linear
programming relaxation. This approach generalizes that of
\cite{DBLP:journals/siamcomp/AustrinGH17}. The following is the pseudocode for establishing
the existence of a solution.
\begin{itemize}

\itemsep=0ex
\item Construct the LP relaxation:
  \begin{itemize}
  
  \itemsep=0ex
  \item For each variable $x_j$ of $\Psi_P$, stipulate that $0 \le v_j \le 1$.
  \item For each clause $P_i(x_{j_1}, \hdots, x_{j_{k_i}})$ in $\Psi_P$, stipulate that $(v_{j_1}, \hdots, v_{j_{k_i}})$ is in the convex hull of the elements of $P_i$. 
  \end{itemize}
\item For each variable $x_j$ of $\Psi_P$.
  \begin{itemize}
  
  \itemsep=0ex
  \item Fix $v_j = 0$ (fix no other variables) and re-solve the LP.
  \item If no solutions, fix $v_j = 1$ and re-solve the LP.
  \item If still no solutions, output `unsatisfiable.'
  \end{itemize}
\item Output `satisfiable.'
\end{itemize}

\begin{rem}
  Since $\Gamma$ is fixed and finite, the size of the LP relaxation is within a constant factor of the size of the instance PCSP. Also, this LP does not require an objective function, since we are only concerned whether the LP has any solution.
\end{rem}

\begin{rem}
  It is worth noting that a different algorithm also exists for the \Maltsev~polymorphism. For each $P_i(x_{j_1}, \hdots, x_{j_{k_i}})$ in $\Psi_P$, write {the} minimal system of linear equations over $\mathbb Z$ such that every element of $P_i$ is a solution (this is known as the \emph{affine hull} of $P_i$). Then, solve this system of linear equations using Gaussian elimination over $\mathbb Z$ (e.g.,~\cite{KannanBachem1979}). Clearly if the system is infeasible, then $\Psi_P$ is unsatisfiable. For any solutions $(v_1, \hdots, v_n)$ to this system, then $(w_1, \hdots, w_n)$ where
\[
w_i = \begin{cases}
  1 & v_i \ge 1\\
  0 & v_i \le 0
\end{cases}
\]
is a solution to $\Psi_Q$.
\end{rem}

\begin{proof}
  Note that the algorithm did not distinguish whether the family $\{\Maj_L\}$ or $\{\Malt_L\}$ were the polymorphisms. The reason the algorithm works, however, differs for these two cases.

  First, assume that $\Psi_P$ is satisfiable. Then, there must exist an integer solution to the linear program. Thus, for each variable $x_j$, there must the LP must be feasible for at least one of $v_j = 0$ or $v_j = 1$. Therefore, the algorithm always correctly reports satisfiable in this case.

  Now, consider the case that $\Psi_Q$ is unsatisfiable. Assume for {the} sake of contradiction, that our algorithm incorrectly reports satisfiable on input $\Psi$. Thus, from our checks, we have that there exists a matrix $M \in \mathbb [0, 1]^{n \times n}$ of solutions (on the columns) such that $M_{i, i} \in \{0, 1\}$ for all $i \in [n]$. Note that we may assume that the entries of $M$ are rational. Furthermore, any convex combination of these $n$ solutions will yield a new solution to the original LP. In other words, for any column vector $v \in [0, 1]^n$, the sum of whose weights is $1$, we have that $Mv$ is also a solution to the LP. Now, we split into cases.

  \textbf{Case 1}, $\Maj_L$ is a polymorphism of $\Gamma$ for all odd $L$.

  We claim that there is $v \in [0, 1]^n$ with sum of coordinates $1$ such that $(Mv)_i \neq 1/2$ for all $i \in [n]$. Consider $w$ with the right properties such that $Mw$ has a minimal number of coordinates equal to $1/2$. If the number of such coordinates is $0$, we are done. Otherwise, consider a coordinate $j$ such that $(Mw)_j = 1/2$. Let $\epsilon = \min\{|(Mw)_i - 1/2|, (Mw)_i \neq 1/2\}/n$. Consider $w' = (1-\epsilon/2)w + (\epsilon/2)e_j$, where $e_j \in \{0, 1\}^n$ has value $1$ in the $j$th coordinate and $0$ everywhere else. Note, then that $|(Mw)_i - (Mw')_i| \le \epsilon/2$ for all $i$, so $Mw'$ will not have any new coordinates equal to $1/2$. Furthermore, since $(Me_j)_j$ is an integer (by construction of $M$), we have that the $(Mw')_j = 1/2 \pm \epsilon / 4 \neq 1/2$ also. Thus, $Mw'$ has fewer coordinates equal to $1/2$, violating the minimality of $w$. Thus, we can find a $v$ such that $(Mv)_i \neq 1/2$ for all $i \in [n]$.

  Thus, now we know that such a $v$ exists, we may consider $\epsilon = \min\{|(Mv)_j - 1/2|\} > 0$. We may perturb $v$ slightly to $v'$ with all of its coordinates rational so that $(Mv')_j \neq 1/2$ for all $j$. Since the coefficients of $M$ are rational, we have that $w = Mv$ has rational entries all not equal to $1/2$. We claim that $x_i^* = \lfloor w_i\rceil$ ($w_i$ rounded to the nearest integer) is a satisfying assignment to $\Psi_Q$. Now, consider any clause $P_i(x_{j_1}^*, \hdots, x_{j_{k_i}}^*)$ of $\Psi_P$, and enumerate the {integral points} $x^1, \hdots, x^{|P|} \in P$. Since $w$ is a rational solution to the LP, we have that there exists $\alpha_1, \hdots, \alpha_{|P|} \in \mathbb Q \cap [0, 1]$ which sum to 1 such that $(w_{j_1}, \hdots, w_{j_{k_i}}) = \alpha_1x^1 + \cdots + \alpha_{|P|}x^{|P|}$. Pick an integer $N \in \mathbb N$ which is a common denominator of $\alpha_1, \hdots, \alpha_{|P|}$. Consider $L = 2N + 1$. Since $\Maj_{L}$ is a polymorphism of $(P, Q)$, we have that that the majority of $2\alpha_1N$ copies of $x^1$, up to $2\alpha_{|P|}N$ copies of $x^N$ and an extra copy of $x^1$ (which has no effect) is in Q. It is easy to verify that this majority is the rounding of the entries of $(w_{j_1}, \hdots, w_{j_{k_i}})$ to the nearest integer. Thus, a satisfying assignment to $\Psi_Q$ exists, a contradiction.
  
  \textbf{Case 2}, $\Malt_L$ is a polymorphism of $\Gamma$ for all odd $L$.

  Let $\hat{w}$ be any rational solution to the LP. Using an argument similar to that in Case 1, we may find $v, w \in [0, 1]^n \cap \mathbb Q^n$ with $\sum_i v_i = 1$ such that $w = Mv$ and $w_i \neq \hat{w}_i$ for all $i \in [n]$ such that $\hat{w}_i \not\in \{0, 1\}$ (otherwise, it may be the case that $w_i = \hat{w}_i$ for all possible $w_i$). We now claim that the following is a satisfying assignment to $\Psi_Q$.
  \[
  \forall i \in [n],\ x_i^* = \begin{cases}0 & w_i < \hat{w_i} \text{ or } w_i = \hat{w_i} = 0\\ 1 & w_i > \hat{w_i} \text{ or } w_i = \hat{w_i} = 1\end{cases}.
  \]
  Consider any clause $P_i(x_{j_1}^*, \hdots, x_{j_{k_i}}^*)$ of $\Psi_P$ and enumerate the {integral points} $x^1, \hdots, x^{|P|} \in P$. Let $\alpha_1, \hdots, \alpha_{|P|}, \hat{\alpha}_1, \hdots, \hat{\alpha}_{|P|} \in [0, 1]\cap \mathbb Q$ be the weights such that $(w_{j_1}, \hdots, w_{j_{k_i}}) = \alpha_1x^1 + \cdots + \alpha_{|P|}x^{|P|}$ and $(\hat{w}_{j_1}, \hdots, \hat{w}_{j_{k_i}}) = \hat{\alpha}_1x^1 + \cdots + \hat{\alpha}_{|P|}x^{|P|}$. Let $N$ be a common denominator of the $\alpha_i$'s and $\hat{\alpha}_i$'s. For $L = 4N + 1$, plug into the odd-indexed entries of $\Malt_L$, $2N\alpha_i$ copies of $x^i$ for all $i \in \{1, \hdots, |P|\}$ and one extra copy of $x^1$ (which will not affect the output of the polymorphism). Into the even-indexed entries plug in $2N\hat{\alpha}_i$ copies of $x^i$ for all $i \in \{1, \hdots, |P|\}$. For each coordinate $\ell \in \{1, \hdots, k\}$, if $w_{j_{\ell}} = \hat{w}_{j_{\ell}} \in \{0, 1\},$ then when computing the $\ell$th coordinate, $\Malt_L$ will have every input equal to $x_{j_{\ell}}^*$ and thus will output that same value, as desired. If $w_{j_{\ell}} < \hat{w}_{j_{\ell}}$, then there will be strictly more $1$s in the even coordinates than in the odd coordinates, so $\Malt_L$ will output $0$ which agrees with our solution $x_{j_{\ell}}^*$. Finally, if $w_{j_{\ell}} > \hat{w}_{j_{\ell}}$, then there will be strictly more $1$s in the odd coordinates than in the even coordinates, so $\Malt_L$ will output $1$ which agrees with our solution $x_{j_{\ell}}^*$. Therefore $\Psi_Q$ is indeed satisfiable, contradiction.
  
  \textbf{End Cases.}
\end{proof}

\begin{rem}
  This only checks whether $\Psi = (\Psi_P, \Psi_Q)$ is satisfiable, but does not find a solution when satisfiable. The proof of correctness may be modified to achieve polynomial-time algorithms for finding a satisfying assignment. See also the subsequent work by the same authors~\cite{BrakensiekGuruswami2019, BrakensiekGuruswami2020} which have a substantial discussion on the differences between decision and search as well as streamlined and generalized the algorithms in this section to a wider range of polymorphisms{.}
\end{rem}

\subsection{Non-idempotent polymorphisms}

Consider a family of promise relations $\Gamma$. If $\Zero_1$ or $\One_1$ is a polymorphism of $\Gamma$, as previously mentioned, it is polynomial-time tractable. Thus, now consider $\Gamma$ non-degenerate. What if $\Pol(\Gamma)$ has none of the idempotent families of polymorphisms mentioned in this section, but it has one of the non-idempotent families (such as $\overline{\Maj}_L$ for all odd $L$)? Then, by Proposition \ref{prop:even-more-idempotence}, the non-idempotent version of this family yields the corresponding idempotent family of polymorphisms of $\Gamma'' = (\neg \Gamma)\cup\{\setzero, \setone\}$. From the previous sections, we then have that $\Gamma''$ is polynomial-time tractable. Therefore, by Lemma \ref{lem:idempotent}, that $\Gamma$ itself is polynomial-time tractable. Hence, we have proved the following.

\begin{thm}\label{thm:alg-result}
  Let $\Gamma$ be a finite family of promise relations. If at least one of $\Zero_L$, $\One_L$, $\AND_L$, $\OR_L$, $\overline{\AND}_L$, or $\overline{\OR}_L$ is a polymorphism of $\Gamma$ for all $L$, or $\Par_L$, $\Maj_L$, $\Malt_L$, $\overline{\Par}_L$, $\overline{\Maj}_L,$ or  $\overline{\Malt}_L$ is a polymorphism of $\Gamma$ for all odd $L$, then $\PCSP(\Gamma)$ is polynomial-time tractable.
\end{thm}

Note that we did not assume that $\Gamma$ is symmetric for our algorithms. That assumption will be incorporated into the $\mathsf{NP}$-hardness arguments.

\section{Classification of Polymorphisms of Folded, Symmetric Promise Relations}\label{sec:classification}

Consider any family $\Gamma$ of finitely many symmetric promise relations which contains the NOT relation. We showed in Section \ref{sec:alg} if the polymorphisms of $\Gamma$ contain any of $\Par_L, \Maj_L, \Malt_L, \overline{\Par}_L, \overline{\Maj}_L,$ or $\overline{\Malt}_L$ for all odd $L$, then $\Gamma$ is polynomial-time tractable. We show in this section that if $\Gamma$ does not have any of these as polymorphisms for all odd $L$, then $\Gamma$'s polymorphisms are rather skewed. Explicitly, we show that every polymorphism $f \in \Pol(\Gamma)$ is a ``$C$-fixing junta.'' This is a key element of the $\mathsf{NP}$-hardness proof in Section~\ref{sec:hardness}.

\subsection{PCSP relaxation}

In order to simplify our proof as well as to illuminate the crucial role of the promise relations, we introduce the notion of \textit{relaxing} a promise relation.

\begin{df}
Let $\Gamma = \{(P_i, Q_i) {: i \in [r]}\}$ be a family of promise relations. We say that another family of promise relations $\Gamma'$ is a \textit{relaxation} of $\Gamma$ if $\Pol(\Gamma) \subseteq \Pol(\Gamma')$.
\end{df}

Intuitively, a larger set of polymorphisms should make the PCSP easier. In Section \ref{app:Galois}, we confirm this by showing that if $\Pol(\Gamma) \subseteq \Pol(\Gamma')$, then there is a polynomial time reduction from $\PCSP(\Gamma')$ to $\PCSP(\Gamma)$. This fact follows from the \textit{Galois correspondence} of polymorphisms and PCSPs. Therefore, since our aim is to demonstrate the $\mathsf{NP}$-hardness of $\PCSP(\Gamma)$, it suffices to show that $\PCSP(\Gamma')$ is $\mathsf{NP}$-hard for some suitable choice of $\Gamma'$ that is a relaxation of $\Gamma$.

The main insight leading to our choice of $\Gamma'$ is our over-arching philosophy that \textit{polymorphisms beget algorithms}. Thus, if we ensure $\Gamma'$ fails to have the polymorphisms which we showed led to polynomial-time algorithms, $\Par_L, \Malt{}_L, \Maj_L$, then $\PCSP(\Gamma')$ should be $\mathsf{NP}$-hard. In the coming subsections, we show exactly which promise relations need to be added to $\Gamma'$ in order to exclude Parity, \Maltsev{}, and Majority, while still including all of the idempotent polymorphisms of $\Gamma$. 

To warm up, here is a claim about such relaxations in the symmetric case. Intuitively, this relation says we can reduce the arity of any symmetric relation in a way which respects the symmetric structure.

\begin{claim}\label{claim:shrink}
Let $(P, Q)$ be a symmetric promise relation of arity $k$. Let $P = \Ham_k(S), Q = \Ham_k(T)$, where $S \subseteq T \subseteq \{0, \hdots, k\}$. Then, {each} idempotent polymorphisms $(P, Q)$ {is a} polymorphism of $(\Ham_{k-1}(S \setminus \{k\}), \Ham_{k-1}(T \setminus \{k\}){)}$.  
\end{claim}

{For this proof, and throughout the paper we shall use the following notation. For $S \subseteq \{1, \hdots, L\}$, we let $e_S \subseteq \{0, 1\}^L$ be such that $(e_S)_i = 1$ if and only if $i \in S$. If $S = \{i\}$ is a single element, we let $e_i = e_S$.}

\begin{proof}
  Let $f \in \Pol(P, Q)$ be any idempotent polymorphism of arity $L$. Consider $x^1, \hdots, x^{k-1} \in \{0, 1\}^L$ such that for all $i \in \{1, \hdots, L\}$, $|(x^1_i, \hdots, x^{k-1}_i)| \in S \setminus \{k\}$. This implies that $|(x^1_i, \hdots, x^{k-1}_i, 0)| \in S$ for all $i$, so since $f \in \Pol(P, Q)$,
  \[
  |(f(x^1), \hdots, f(x^{k-1}), f(0\hdots 0))| \in T
  \]
  Thus, since $f$ is idempotent, $|(f(x^1), \hdots, f(x^{k-1}))| \in T\setminus \{k\}$. Thus, $f \in \Pol(\Ham_{k-1}(S \setminus \{k\}), \Ham_{k-1}(T \setminus \{k\}).$
\end{proof}

Let $P$ be any relation of arity $k$, and let $S \subseteq \{1, \hdots, k\}$ be any subset. Then, define
\[\flip_S(P) = \{y \in \{0, 1\}^k\: :\: y\oplus e_S \in P\}.\]
Note that $\neg P = \flip_{[k]}(P)$.

\begin{claim}\label{claim:flip}
Let $(P, Q)$ be a promise relation of arity $k$, and let $S \subseteq \{1, \hdots, k\}$. Then, $(P, Q)$ and $(\flip_S(P), \flip_S(Q))$ have identical folded polymorphisms.
\end{claim}
\begin{proof}
Consider any $f \in \Pol(P, Q)$ of arity $L$ which is folded. Pick $x^1, \hdots, x^k \in \{0, 1\}^L$ such that $(x^1_j, \hdots, x^k_j) \in \flip_S(P)$ for all $j \in \{1, \hdots, L\}$. Then, consider $y^1, \hdots, y^k$ such that $y^i = \neg x^i$ if $i \in S$ and $y^i = x^i$ otherwise. Then, for all $j \in \{1, \hdots, L\}$, $(y^1_j, \hdots, y^k_j) \in P$. Thus, $(f(y^1), \hdots, f(y^k)) \in Q$. Due to folding, we have that $(f(x^1), \hdots, f(x^k)) \in \flip_S(Q)$. Thus, the folded polymorphisms of $(P, Q)$ are polymorphisms of $(\flip_S(P), \flip_S(Q))$. By a symmetric argument, we may deduce that the folded polymorphisms of $(P, Q)$ and $(\flip_S(P), \flip_S(Q))$ are identical.
\end{proof}

\begin{rem}
Note that unless $S = \{\}$ or $S = \{1, \hdots, k\}$, then $\flip_S(P)$ and $\flip_S(Q)$ may not be symmetric if $P$ and $Q$ were originally symmetric. That said, in most applications, $S$ will be one of these symmetry-preserving choices.
\end{rem}

We can combine these two claims to get a natural corollary. This result tells us that we can shift down the Hamming weights of a symmetric, folded promise relation.

\begin{claim}\label{claim:shrink2}
  Let $(P, Q)$ be a symmetric promise relation of arity $k$. Let $P = \Ham_k(S), Q = \Ham_k(T)$, where $S \subseteq T \subseteq \{0, \hdots, k\}$. Then, the idempotent, folded polymorphisms of $(P, Q)$ are polymorphisms of $(\Ham_{k-1}(\{{\ell: \ell \ge 0, \ell + 1 \in S}\}), \Ham_{k-1}({\{\ell: \ell \ge 0, \ell + 1 \in T\}}))$.  
\end{claim}
\begin{proof}
Apply Claim~\ref{claim:flip} on $\{1,\hdots, k\}$ to reduce the idempotent, polymorphisms of $(P, Q)$ to $(\Ham_k(\{k - \ell : \ell \in S\}, \Ham_k(\{k - \ell : \ell \in T\})$. Then, we apply Claim \ref{claim:shrink} to reduce further to $(\Ham_{k-1}(\{k - \ell : \ell \in S\} \cap \{0, \hdots, k- 1\}), \Ham_{k-1}(\{k - \ell : \ell \in T\} \cap \{0, \hdots, k - 1\}))$. Finally, we use Claim \ref{claim:flip} again to reduce the idempotent, folded polymorphisms of $(P, Q)$ to $(\Ham_{k-1}(\{{\ell: \ell \ge 0, \ell + 1 \in S}\}), \Ham_{k-1}(\{{\ell: \ell \ge 0, \ell + 1 \in T}\})),$ as desired.
\end{proof}

In the following sections, we will repeatedly use the claims to the reduce the $(P, Q)$ of our $\Gamma$ to some simpler promise relations for which we can analyze the folded, idempotent polymorphisms.

\subsubsection{\Maltsev{}-excluding relaxation}
\begin{lem}\label{lem:no-maltsev}
  Let $\Gamma$ be a symmetric, folded, idempotent family of promise relations such that $\Malt_L \not\in \Pol(\Gamma)$ for some odd positive integer $L$, then $\Gamma' = \{(P, Q)\}$ is a relaxation of $\Gamma$, in which either
  \begin{align*}
  P &= \Ham_k(\{1\}), &Q &= \Ham_k(\{0, 1, \hdots, k-2, k\}), k \ge 3, \text{ or}\\
  P &= \Ham_k(\{0, b\}), &Q &= \Ham_k(\{0, \hdots, k-1\}), k \ge 2, b\in \{1, \hdots, k-1\}.
  \end{align*}
\end{lem}
\begin{proof}
  As $\Malt_L \not\in \Pol(\Gamma)$, there is $(P, Q) \in \Gamma$ such that $\Malt_L \not\in \Pol(P, Q)$. Define ${\Malt}(P) = \bigcup_{L \in \mathbb N, \text{odd}} {\Malt_L}(P)$. Since $\Malt_L \not\in \Pol(P, Q)$ for some odd $L$, we have that ${\Malt}(P) \not\subseteq Q.$ We claim the following.
  \begin{claim}\label{claim:Maltsev-brute-force}
    Consider $k \ge 1$, then
    \begin{enumerate}
    \item ${\Malt}(\Ham_k(\{0\})) = \Ham_k(\{0\})$
    \item ${\Malt}(\Ham_k(\{k\})) = \Ham_k(\{k\})$
    \item ${\Malt}(\Ham_k(\{0, k\})) = \Ham_k(\{0, k\})$
    \item ${\Malt}(\Ham_k(\{\ell\})) = \Ham_{k}(\{1, \hdots, k-1\}),\ k \ge 2,\ \ell \in \{1, \hdots, k - 1\}$
    \item ${\Malt}(\Ham_k(\{\ell_1, \ell_2\})) = \{0, 1\}^k,\ k \ge 2, \{\ell_1, \ell_2\} \neq \{0, k\}, \ell_1 \neq \ell_2$
    \end{enumerate}
    \begin{proof}
    Facts 1-3 are easy to verify since $\Malt_L$ is idempotent for all odd $L$.

    For Fact 4, consider $\ell, \ell' \in \{1, \hdots, k - 1\}$. We claim that $\Ham_k(\{\ell'\}) \subseteq {\Malt}(\Ham_k(\{\ell\}))$. Pick $L = 2\ell'(k-\ell') + 1$. It suffices to pick $x^1, \hdots, x^L \in \Ham_k(\{\ell\})$ such that $\Malt_L(x^1, \hdots, x^L) = (1, \hdots, 1, 0, \hdots, 0),$ where the output has Hamming weight $\ell'$. Let $x^1 = (1, \hdots, 1, 0, \hdots, 0)$, of Hamming weight $\ell$, and let $x^2 = (0, \hdots, 0, 1,\hdots, 1)$, of Hamming weight $\ell$. Let $x^1, x^3, \hdots, x^{L - 2}$ be all possible permutations of $x^1$ in which the first $\ell'$ coordinates are cyclically shifted and the last $k - \ell'$ coordinates are cyclically shifted. There may be repetition, but each repetition should appear an equal number of times. Likewise, let $x^2, x^4, \hdots, x^{L - 1}$ be the same kind of permutations but of $x^2$. Let $x^L = 1$. It is easy to verify that
    \begin{align*}
    j \in \{1, \hdots, \ell'\}, \sum_{i=1, \text{odd}}^{L-2} x^i_j &= (k-\ell')\min(\ell, \ell')\\
    j \in \{\ell'+1, \hdots, k\}, \sum_{i=1, \text{odd}}^{L-2} x^i_j &= \ell'\max(0, \ell-\ell')\\
    j \in \{1, \hdots, \ell'\}, \sum_{i=1,\text{even}}^{L-1} x^i_j &= (k-\ell')\max(0, \ell+\ell'-k)\\
    j \in \{\ell'+1, \hdots, k\}, \sum_{i=1, \text{even}}^{L-1} x^i_j &= \ell'\min(\ell, k-\ell')
    \end{align*}
    Thus,
    \begin{align*}
      j\in \{1, \hdots, \ell'\}, \sum_{i=1}^L (-1)^{i-1}x^i_j &= (k-\ell')(\min(\ell, \ell') - \max(0, \ell+\ell'-k)) + x^L_j \ge (k-\ell') + x^L > 0\\
      j\in \{\ell'+1 \hdots, k\}, \sum_{i=1}^L (-1)^{i-1}x^i_j &= \ell'(\max(0, \ell-\ell') - \max(\ell, k-\ell')) + x^L_j < -\ell' + x^L_j \le 0.
    \end{align*}
    Therefore, $\Malt_L(x^1, \hdots, x^L) \in \Ham_k(\{\ell'\})$, as desired. By Proposition \ref{prop:sym}, $\Ham_k(\{\ell'\}) \subseteq {\Malt}(\Ham_k(\{\ell\}))$, for all $\ell, \ell' \in \{1, \hdots, k-1\}$, as desired.

    Now, to finish Fact 4, we seek to show that $\Ham_k(\{0\}) = \{(0, \hdots, 0)\} \not\subseteq \Ham_k(\{\ell\})$. Assume for {the} sake of contradiction, there exists $L$ odd and $x^1, \hdots, x^L \in \Ham_k(\{\ell\})$ such that $\Malt_L(x^1, \hdots, x^L) = (0, \hdots, 0)$. Then, we have that for all $i \in \{1, \hdots, k\}$, $\sum_{j=1}^L (-1)^{j-1}x^j_i\le 0$. Summing over all $i$, we have that $0 \ge \sum_{j=1}^L (-1)^{j-1}\sum_{i=1}^kx^j_i = \sum_{j=1}^L (-1)^{j-1}\ell = \ell$, a contradiction. Likewise, if $\Malt_L(x^1, \hdots, x^L) = (1, \hdots 1)$. We would have that $k \le \sum_{j=1}^L (-1)^{j-1}\sum_{i=1}^L x^j_i = \sum_{j=1}^L(-1)^{j-1}\ell = \ell$, which is also a contradiction. Thus, we have shown fact 4.

    For Fact 5, since we know that $\{\ell_1, \ell_2\} \neq \{0, k\}$, we know that at least one of $\ell_1$ and $\ell_2$ is strictly between $1$ and $k-1$. Thus, by fact 4, $\Ham_{k}(\{1, \hdots, k-1\}) \subseteq {\Malt}(\Ham_k(\{\ell_1, \ell_2\})).$ Therefore, it suffices to prove that $(0, \hdots, 0), (1, \hdots, 1) \in {\Malt}(\Ham_k(\{\ell_1, \ell_2\}))$. Assume that $\ell_1 < \ell_2$ and consider $L = 4k + 1$. To show $(0, \hdots, 0) \in {\Malt}(\Ham_k(\{\ell_1, \ell_2)\})$, pick $x^1, \hdots, x^{4k+1}$ such that $x^j$ has Hamming weight $\ell_1$ when $j$ is odd and Hamming weight $\ell_2$ when $j$ is even. Let $x^1 = (1, \hdots, 1, 0, \hdots, 0)$, and let $x^3, x^5, \hdots$ be successive cyclic shifts. Likewise, let $x^2 = (1, \hdots, 1, 0, \hdots, 0)$ (with the appropriate number of 1s), and let $x^4, \hdots$ be successive cyclic shifts. Then, it is easy to see that for all $i \in \{1, \hdots, k\}$, we have that $\sum_{j=1}^{4k+1} (-1)^{i-1}x^j_i$ is $2\ell_1 - 2\ell_2 < 0$ or $2\ell_1 - 2\ell_2 + 1 < 0$ (because we have $2k+1$ odd-indexed terms but $2k$ even-indexed terms). Thus, $\Malt_{4k+1}(x^1, \hdots, x^{4k+1}) = (0, \hdots, 0)$. If we do the same construction but swap $\ell_1$ and $\ell_2$, we would then have that $2\ell_1 - 2\ell_2, 2\ell_1 - 2\ell_2 + 1 > 0$, so $\Malt_{4k+1}(x^1, \hdots, x^{4k+1}) = (1, \hdots, 1)$. Thus, Fact 5 is shown.
  \end{proof}

  \end{claim}

  Now, let $a \in \{0, \hdots, k\}$ be such that $\Ham_k(\{a\}) \subseteq {\Malt}(P)$ but $\Ham_k(\{a\}) \not \subseteq Q$. Such an $a$ must exist by Proposition \ref{prop:sym}. Note that since $P \subseteq Q$, we must have that $\Ham_k(\{a\}) \not \subseteq P$. We divide the remaining analysis into two cases.

  \textbf{Case 1,} $a \in \{1, \hdots, k - 1\}$.

  Then, by Fact 3 of the claim, there must exists $\ell \in \{1, \hdots, k - 1\}$ such that $\Ham_k(\{\ell\}) \subseteq P$ (else, $P \subseteq {\Ham_k(\{0, k\})}$ and then $\Ham_k(\{a\}) \not\subseteq {\Malt}(P)$.) Let $P' = \Ham_k(\{\ell\})$ and let $Q' = \Ham_k(\{0, \hdots, k\}-\{a\})$. Since $P' \subseteq P \subseteq Q \subseteq Q'$, every polymorphism of $(P, Q)$ is a polymorphism of $(P', Q')$. Furthermore, by Fact 4, $\Ham_k(\{a\}) \subseteq {\Malt}(P')$ but $\Ham_k(\{a\}) \not\subseteq Q'$. Thus, $(P', Q')$ does not admit $\Malt_{L'}$ as a polymorphism for some $L'$.

  {We may assume without loss of generality that $\ell < a$. Else we may apply Claim~\ref{claim:flip} with $S = [k]$.} Let $k' = {a} + 1$. From Claim \ref{claim:shrink}, applied $k - k'$ times, we have that all of the idempotent polymorphisms of $(P', Q')$ are idempotent polymorphisms of
  \[
  P^{(2)} = \Ham_{k'}(\{\ell\}), Q^{(2)} = \Ham_{k'}(\{0, \hdots, k'\} - \{a\}).
  \]
  Likewise, applying Claim~\ref{claim:shrink2} ${\ell} - 1$ times, all of the folded polymorphisms of $P^{(2)}, Q^{(2)}$ are polymorphisms of
  \[
  P^{(3)} = \Ham_{k''}(\{\ell'\}), Q^{(3)} = \Ham_{k''}(\{0, \hdots, k''\} - \{a'\})
  \]
  where $k'' = {a - \ell} + 2 \ge 3$. {Note that} $\ell' = 1$ and $a' = k''-1${. Thus,} the idempotent, folded polymorphisms of $\Gamma$ are polymorphisms of
  \[
  P^{(4)} = \Ham_{k''}(\{1\}), Q^{(4)} = \Ham_{k''}(\{0, \hdots, k''-2, k''\}).
  \]

  \textbf{Case 2,} $a \in \{0, k\}$.

  Without loss of generality, we may assume that $a = 0$. Otherwise, we may replace $(P, Q)$ with $(\flip_{[k]}(P), \flip_{[k]}(Q))$, which preserves the folded, idempotent polymorphisms of $\Gamma$. Since $\Ham_{k}(\{0\}) \subseteq {\Malt}(P)$ but $\Ham_{k}(\{0\}) \not\subseteq P \subseteq Q$, we must be in Fact 5 of {Claim~\ref{claim:Maltsev-brute-force}}. That is, there must be $\ell_1, \ell_2 \in \{0, \hdots, k\}$ distinct and ${\{\ell_1, \ell_2\} \neq }\{0, k\}$ such that $\Ham_k(\{\ell_1, \ell_2\}) \subseteq P$. Like in Case 1, relax $(P, Q)$ to $P' = \Ham_k(\{\ell_1, \ell_2\})$ and $Q' = \Ham_k(\{1, \hdots, k\})$. Let $k' = \max(\ell_1, \ell_2)$, and apply Claim \ref{claim:shrink} $k - k'$ times to yield
  \[
  P^{(2)} = \Ham_{k'}(\{\min(\ell_1, \ell_2), k'\}), Q^{(2)} = \Ham_{k'}(\{1, \hdots, k'\}).
  \]
  Then, applying Claim~\ref{claim:flip} with $S = \{1, \hdots, k'\}$, we get that
  \[
  P^{(3)} =  \Ham_{k'}(\{0, b\}), Q^{(3)} = \Ham_{k'}(\{0, \hdots, k'-1\}), b = k' - \min(\ell_1, \ell_2) \in \{1, \hdots, k'-1\}, k'\ge 2
  \]
  has as polymorphisms the folded, idempotent polymorphisms of $\Gamma$, as desired.
  
  \textbf{End Cases.}
\end{proof}

\subsubsection{Majority-excluding relaxation}
\begin{lem}\label{lem:no-majority}
  Let $\Gamma$ be a symmetric, folded, idempotent family of promise relations such that $\Maj_L \not\in \Pol(\Gamma)$ for some odd positive integer $L$, then $\Gamma' = \{(P, Q)\}$ is a relaxation of $\Gamma$, in which either
  \begin{align*}
  P &= \Ham_k(\{(k+1)/2\}), &Q &= \Ham_k(\{0, 1, \hdots, k-1\}), \text{ ($k\ge 3$ odd), or}\\
  P &= \Ham_k(\{1, k\}), &Q &= \Ham_k(\{0, 1, \hdots, k\} - \{b\}), k \ge 3, b \in \{2, \hdots, k - 1\}.
  \end{align*}
\end{lem}
\begin{proof}
  The proof proceeds in a similar manner to Lemma \ref{lem:no-maltsev}. Define ${\Maj}(P) = \bigcup_{L \in \mathbb N, \text{odd}}{\Maj_L}(P)$. We begin with the analogue of Claim~\ref{claim:Maltsev-brute-force} for the Majority operation.
  \begin{claim}\label{claim:Majority-brute-force}
    Consider $k \ge 1$. If $P \subseteq \Ham_k(\{0, k\})$, then ${\Maj}(P) = P$. Otherwise, if $P = \Ham_k(S)$ is symmetric but $S \setminus \{0, k\}$ is nonempty, then
    \[
    {\Maj}(P) = \Ham_k(\{0, \hdots, k\} \cap \{2(\min S) - k + 1, \hdots, 2(\max S) - 1\}).
    \]
  \end{claim}

      \begin{proof}
    As before, the case $S \subseteq \{0, k\}$ is easy.

    We first, show that ${\Maj}(P) \subseteq \Ham_k(\{0, \hdots, k\} \cap \{2\min S - k + 1, \hdots, 2\max S - 1\}).$ For any $b \in \{0, \hdots, k\}$, such that $\Ham_k(\{b\}) \subseteq {\Maj}(P)$, there is $L$ odd and $x^1, \hdots, x^L \in P$ such that $\Maj_L(x^1, \hdots, x^L)$ has Hamming weight $b$. Assume without loss of generality the coordinates equal to $1$ are the first $b$ ones. Thus, we have that for all $i \in \{1, \hdots, b\}$, $\sum_{j=1}^L x^j_i \ge (L+1)/2.$. Thus,
    \[
    \sum_{i=1}^{k}\sum_{j=1}^L x^j_i \ge b(L+1)/2.
    \]
    Thus, by the pigeonhole principle, there is some $x^j$ such that its Hamming weight is at least $b(L+1)/(2L) \le \max S$. Thus, $b \le 2L\max S/(L+1) < 2\max S$. Therefore, $b \le 2\max S -1$, as desired. Using the fact that $\sum_{j=1}^L x^j_i \le (L-1)/2$ for all $i \in \{b +1, \hdots, k\}$, we have that some $x^j$ has Hamming weight at most $b + (k-b)(L-1)/(2L) \ge \min S$ (we add $b$ since all of the first $b$ coordinates may be $1$s). Thus, $b \ge 2L\max S/(L - 1) - k(L-1)/(L+1) > 2\max S - k$. Therefore, $b \ge 2\max S + 1 - k$, as desired. Thus, ${\Maj}(P) \subseteq \Ham_k(\{0, \hdots, k\} \cap \{2\min S - k + 1, \hdots, 2\max S - 1\})$.

    Now, we show the reverse direction, that every $b \in \{0, \hdots, k\} \cap \{2\min S - k + 1, \hdots, 2\max S - 1\}$ can be obtained as a Hamming weight. Assume that there is $\ell \in S \cap \{1, \hdots, k-1\}$ (so $k \ge 2$). For ease of notation, let $s = \min S, t = \max S$, so $s \le \ell \le t$. To start, we show if $b \in \{0, \hdots, k\} \cap \{\ell,  \hdots, 2\max S - 1\}$, then $\Ham_k(\{b\}) \subseteq {\Maj}(P)$.

    First, if $b \ge \max S$, consider $L = 2b+1$. We now seek to pick $x^1, \hdots, x^L$ of Hamming weight $\max S$ such that $\Maj_L(x^1, \hdots, x^L)$ has Hamming weight $b$. Let $x^1 = (1, \hdots, 1, 0, \hdots, 0)$ of the suitable Hamming weight, and for all $j \ge 2$, let $x^j$ be the cyclic shift of $x^{j-1}$ in the first $b \ge \max S$ coordinates. For all $i \in \{b+1, \hdots, k\}$, $\sum_j x^j_i = 0$, so $\Maj_L(x^1_i, \hdots, x^L_i) = 0$. For all $i \in \{1, \hdots, b\}$, $\sum_j x^j_i = 2\max S + x^L_j \ge b + 1 + x^L_j > L/2$. Thus, $\Maj_L(x^1_j, \hdots, x^L_j) = 1$, so $\Maj_L(x^1, \hdots, x^L)$ has Hamming weight $b$.

    Otherwise, if $b \in \{\ell, \ell +1, \hdots, \max S - 1\}$, consider now $L = 2b-1 \ge 1$. Let $x^1 = \cdots = x^{b-1} = (1, \hdots, 1, 0, \hdots, 0)$, with Hamming weight $\max S$. Let $x^b = (1, \hdots, 1, 0, \hdots, 0)$ of Hamming weight $\ell$, and let $x^{b+1}, \hdots, x^{2b-1}$ be $x^b$ except that the first $b$ coordinates are cyclically shifted. If $i \in \{1, \hdots, b\}$ then $\sum_i x^j_i = b-1 + \ell \ge b > L/2$ since $\ell \ge 1$. If $i \in \{b+1, \hdots, k\}$, then $\sum_i x^j_i\le b-1 < L/2$. Thus, $\Maj_L(x^1, \hdots, x^L)$ has Hamming weight $b$, as desired. Thus, $\Ham_k(\{0, \hdots, k\} \cap \{\ell, \hdots, 2\max S - 1\}) \subseteq {\Maj}(P)$.
    
    By an analogous argument, we may show that $\Ham_k(\{0, \hdots, k\} \cap \{2\min S - k + 1, \hdots, \ell\}) \subseteq {\Maj}(P)$. A simple route to this is reversing the notions of $0$ and $1$ in our previous construction.
  \end{proof}

  Consider $b \in \{0, \hdots, k\}$ such that $\Ham_{k}(\{b\}) \subseteq {\Maj}(P)$ but $\Ham_k(\{b\}) \not\subseteq Q, P$. From Claim~\ref{claim:Majority-brute-force}, $P \setminus \Ham_k(\{0, k\})$ must be nonempty. Thus, there is $\ell \in \{1, \hdots, k-1\}$ such that $\Ham_k(\{\ell\}) \subseteq P$. We may assume without loss of generality that $\ell < b$ as $\ell \neq b$ and we can apply Claim~\ref{claim:flip} to $(P, Q)$ to get $(\flip_{[k]}(P),\flip_{[k]}(Q))$, which does not change the folded, idempotent polymorphisms.

  Let $S \subseteq \{0, \hdots, k\}$ be such that $P = \Ham_k(S)$. Since $\Ham_k(\{b\}) \subseteq {\Maj}(P)$ and $\ell < b$, we have by the claim that $\Ham_k(\{b\}) \subseteq {\Maj}(\Ham_k(\{\ell, \max S\}))$. Thus, we can relax to $P' = \Ham_k(\{\ell, \max S\})$ and $Q' = \Ham_k(\{0, \hdots, k\} - \{b\})$ while still preserving the idempotent, folded polymorphisms of $\Gamma$. We again diverge into two cases.

  \textbf{Case 1,} $b > \max S$.

  We may relax to
  \[
  P^{(2)} = \Ham_{k}(\{\max S\}), Q^{(2)} = \Ham_{k}(\{0, \hdots, k\} - \{b\}).
  \]
  Let $k' = b$, and apply Claim \ref{claim:shrink} $k - k'$ times to relax the folded, idempotent polymorphisms to
  \[
  P^{(3)} = \Ham_{k'}(\{\max S\}), Q^{(3)} = \Ham_{k'}(\{0, \hdots, k'-1\}).
  \]
  Recall that $\max S < k' = b \le 2\max S - 1$ {by Claim~\ref{claim:Majority-brute-force}}.  Thus, $k'' = 2(k' - \max S) + 1 \le k'$. Applying Claim \ref{claim:shrink2} $k' - k''$ times, we then get that the folded, idempotent polymorphisms of $\Gamma$ are also polymorphisms of
  \[
  P^{(4)} = \Ham_{k''}(\{(k''+1)/2\}), Q^{(4)} = \Ham_{k''}(\{0, \hdots, k''-1\}), k'' \ge 3.
  \]
  This establishes the first case of the lemma.

  \textbf{Case 2,} $\ell < b < \max S$.

  Letting $k' = \max S$ and applying Claim \ref{claim:shrink} $k - k'$ times, we get that the idempotent, folded polymorphisms of $\Gamma$ {are also polymorphisms of}
  \[
  P^{({2})} = \Ham_{k'}(\{\ell, k'\}), Q^{({2})} = \Ham_{k'}(\{0, \hdots, k'\} - \{b\}).
  \]
  Now, consider $k'' = k' - \ell + 1$, and apply Claim \ref{claim:shrink2} $k' - k''$ times to get that the idempotent, folded, polymorphisms of $\Gamma$ {are also polymorphisms of}
  \[
  P^{({3})} = \Ham_{k''}(\{1, k''\}), Q^{({3})} = \Ham_{k''}(\{0, \hdots, k''\} - \{b'\}), b' \in \{2, \hdots, k''-1\}.
  \]
  Note that $k'' \ge 3$; therefore the second case of the lemma is also fully established.
  
  \textbf{End Cases.}
\end{proof}

\begin{rem}
{We do not have a section specific to Parity-avoiding promise relations. It is convenient enough for us to reason about excluding Parity along with another familiy of polymorphisms (e.g., Lemma~\ref{lem:noMalt-noPar}).}
\end{rem}
\subsection{Idempotent case}

We now seek to establish that if a symmetric, idempotent, folded family of promise relations $\Gamma$ avoids $\Par_{L_1}, \Malt_{L_2}, \Maj_{L_3}$ as polymorphisms for some odd $L_1, L_2, L_3$, then the polymorphisms are $C$-fixing for some suitable constant $C(\Gamma)$. Note that this $C$ may depend on $L_1, L_2, L_3$, but if we pick $L_1, L_2, L_3$ to be minimal, then these also depend only on $\Gamma$. Our first step is to establish the following lemma in additive combinatorics.

\begin{lem}\label{lem:sum-game}
  Let $S_0, S_1 \subseteq \mathbb Z_{\ge 0}$ such that $0 \in S_0$ and $1 \in S_1$. Assume that there exists a positive integer $n$ such that for all $a \in S_0$ and $b_1, \hdots, b_{n} \in S_1$ (not necessarily distinct)
  \begin{align}
    b_1 + \cdots + b_n &\in S_0 \label{eq:sum-0}\\
    a + b_1 + \cdots + b_{n-1} &\in S_1. \label{eq:sum-1}
  \end{align}
  If $n$ is odd, then there is $A(n) \in \mathbb Z_{\ge 0}$ such that $A(n) \in S_0 \cap S_1$. Otherwise, if $n$ is even, there is $d(n) \in \mathbb Z_{\ge 0}$ such that $S_0$ contains all even integers at least $d(n)$ and $S_1$ contains all odd integers at least $d(n)$.
\end{lem}

\begin{proof}
  If $n = 1$, then by (\ref{eq:sum-1}), we have that $0 \in S_1$. Thus, we can set $A(1) = 0.$ Now assume $n \ge 2$. We can easily deduce the following facts
  \begin{align}
    \forall x \in S_0{,\ }&x + n - 1 \in S_1&\text{ ($a = x$ and $b_1, \hdots, b_{n-1} = 1$ in (\ref{eq:sum-1}))} \label{eq:0->1}\\
    \forall y \in S_1{,\ }&y + n - 1 \in S_0&\text{ ($b_1 = y$ and $b_2, \hdots, b_{n-1} = 1$ in (\ref{eq:sum-0}))} \label{eq:1->0}\\
    \forall y \in S_1{,\ }&y + n - 2 \in S_1&\text{ ($a = 0$, $b_1 = y$, and $b_2, \hdots, b_{n-1} = 1$ in (\ref{eq:sum-1}))} \label{eq:1->1}
  \end{align}
  In particular, we may deduce that
  \begin{align}
    \forall x \in S_0{,\ }&x + 2n-2 \in S_0&\text{ (\ref{eq:0->1} and \ref{eq:1->0})}\label{eq:0->0,1}\\
    \forall x \in S_0{,\ }&x + 3n-4 \in S_0&\text{ (\ref{eq:0->1}, \ref{eq:1->1}, and \ref{eq:1->0})}\label{eq:0->0,2}\\
    \forall y \in S_1{,\ }&y + n-2 \in S_{{1}}&\text{ (\ref{eq:1->1})} \nonumber\\
    \forall y \in S_1{,\ }&y + 2n-2 \in S_1&\text{ (\ref{eq:1->0} and \ref{eq:0->1})}\label{eq:1->1,2}
  \end{align}

  Note that if $n \ge 3$ is odd, then $\gcd(2n-2, 3n-4)= 1$. Therefore, by (\ref{eq:0->0,1}) and (\ref{eq:0->0,2}) and that $0 \in S_0$, we may deduce by the standard analysis to the Frobenius coin problem (e.g., \cite{brauer1962problem}) (colloquially known as the Chicken McNugget Theorem) that $S_0$ contains all sufficiently large positive integers. Likewise, since $\gcd(n-2, 2n-2) = 1$ and (\ref{eq:1->1}) and (\ref{eq:1->1,2}) hold, we have that $S_1$ contains all sufficiently large positive integers. Hence, there must exist $A(n) \in \mathbb N$ such that $A(n) \in S_0 \cap S_1$.

  If $n \ge 2$ is even, then $\gcd(2n-2, 3n-4) = \gcd (n-2, 2n-2) = 2$. Since $0 \in S_0$, and $1 \in S_1$, we may then deduce by the same theorem that $S_0$ will contain all sufficiently large even numbers and that $S_1$ will contain all sufficiently large odd numbers. Thus, we may select $d(n)$ accordingly.
\end{proof}

\begin{rem}
  Consider the modification that there is some positive integer $N$ such that $\max S_0, \max S_1 \le N$ with the stipulation that (\ref{eq:sum-0}) and (\ref{eq:sum-1}) only apply when the sums are {at} most $N$. The theorem still holds as long as{, when $n$ is odd, the original $A(n)$ is at most $N$, or when $n$ is even, the original $d(n)$ is at most $N$.} {These follow from the fact that $A(n)$ and $d(n)$ are} independent of ${S_1, S_2}$.
\end{rem}

With this established, we can now deduce significant structural properties of the polymorphisms of \Maltsev{} and Parity-avoiding families of promise relations. These {arguments have} connections to those in \cite{DBLP:journals/siamcomp/AustrinGH17}, but differ significantly in details.

\begin{lem}\label{lem:noMalt-noPar}
  Let $\Gamma$ be a symmetric, folded, idempotent family of promise relations such that $\Par_{L_1}, \Malt_{L_2} \not\in \Pol(\Gamma)$ for some odd positive integers $L_1, L_2$. Then, there exists $c(\Gamma) \in \mathbb N$ such that for all $L \in \mathbb N$ and for all $f : \{0, 1\}^L \to \{0, 1\} \in \Pol(\Gamma)$,
  \[
  |\{i \in \{1, \hdots, L\}\: : \: f(e_i) = 1\}| \le c(\Gamma).
  \]
\end{lem}
\begin{proof}
  Fix $f \in \Pol(\Gamma)$ of arity $L$, and let $A = \{i \in \{1, \hdots, L\} \: : \: f(e_i) = 1\}$. Assume for {the} sake of contradiction that $|A|$ {grows} arbitrarily large. Define $S_0, S_1 \subseteq \{0, 1, \hdots, |A|\}$ as follows.
  \[
  S_i = \{j \::\: \text{for all $T \subseteq A$ of size $j$, $f(e_{T}) = i$}\},\ i \in \{0, 1\}.
  \]
  
  We seek to show that there exists $n(\Gamma)$ for which (\ref{eq:sum-0}) and (\ref{eq:sum-1}) hold, so that that we may invoke Lemma \ref{lem:sum-game} on $S_0$ and $S_1$. 
  {Either, we} have some $f$ such that $S_0 \cap S_1$ is nonempty, an immediate contradiction. {Or, $f$ will have structure similar to that of parity, which we can then use to contradict that} $\Par_{L_1} \not\in \Pol(\Gamma)${.}

  To achieve the first goal, which is to show that $n(\Gamma)$ exists which satisfies (\ref{eq:sum-0}) and (\ref{eq:sum-1}), we utilize Lemma \ref{lem:no-maltsev} to deduce a symmetric $(P, Q)$, independent of $L$, such that $f \in \Pol(P,Q)$. {Let $k$ be the arity of $P$ and $Q$.} The proof now proceeds into two cases.

  \textbf{Case 1,} $k \ge 3$, $P = \Ham_k(\{1\}), Q = \Ham_k(\{0, \hdots, k-2, k\})$.

  Let $n = k - 1$. For any $b_1, \hdots, b_{k-1} \in S_1$ such that $b_1 + \cdots + b_{k-1} \le |A|$, consider any $T \subseteq A$ of size $b_1 + \hdots + b_{k-1}$. Partition $T = T_1 \cup T_2 \cup \cdots \cup T_{k-1}$ such that $|T_i| = b_i$ for all $i$. {Let $T_k := [L] \setminus T$.} Consider the $k$-tuple {of $L$-tuples}
  \[(e_{T_1}, e_{T_2}, \hdots, e_{T_{k-1}}, e_{{T_k}}){.}\]
  For every $i \in \{1, \hdots, L\}$, {there is exactly one $j \in [k]$  with $i \in T_j$.} Thus, since $f \in \Pol(P, Q)$, we have that
  \[
  (f(e_{T_1}), f(e_{T_2}), \hdots, f(e_{T_{k-1}}), f(e_{{T_k}})) 
  \]
  has Hamming weight not equal to $k - 1$. Since $f(e_{T_i}) = 1$ for all $i {\in [k-1]}$, we must then have that $f(e_{{T_k}}) = 1$. Since $f$ is folded {and $T = [L] \setminus T_k$}, we can thus deduce that $f(e_{T}) = 0$, as desired. Since the choice of $T\subseteq A$ {is} arbitrary except for size, we have that $b_1 + \cdots + b_{k-1} \in S_0$, so (\ref{eq:sum-0}) holds.

  Now, consider any $a \in S_0$ and $b_1, \hdots, b_{k-2} \in S_1$ such that $a + b_1 + \cdots + b_{k-2} \le |A|$. Again, consider any $T \subseteq A$ of size $a + b_1 + \cdots + b_{k-2}$. Partition $T = T_0 \cup T_1 \cup \cdots \cup T_{k-2}$ such that $|T_0| = a$ and $|T_i| = b_i$ for all other $i$. {Let $T_{k-1} := [L] \setminus T$} Note again that the ${k}$-tuple {of $L$-tuples}
  \[(e_{T_0}, \hdots, e_{{T_{k-2}}}, e_{{T_{k-1}}})\]
  has for every $i \in \{1, \hdots, L\}$ has exactly one {$j \in \{0, 1, \hdots, k-1\}$ with $j \in T_i$}. Thus, we may again deduce that since $f \in \Pol(P, Q)$,
  \[(f(e_{T_0}), \hdots, f(e_{{T_{k-2}}}), f(e_{{T_{k-1}}}))\]
  has Hamming weight not equal to $k-1$. Since exactly $k-2$ of the first $k-1$ entries are equal to $1$, we must have $f(e_{{T_{k-1}}}) = 0$. Thus, $f(e_{T}) = 1$. {Since the choice of $T\subseteq A$ is arbitrary except for size,} $a + b_1 + \cdots + b_{k-2} \in S_1$, so (\ref{eq:sum-1}) also holds, as desired.

  \textbf{Case 2,} $k\ge 2{,}$ $P = \Ham_k(\{0, b\}), Q = \Ham_k(\{0, \hdots, k-1\})$, $b \in \{1, \hdots, k-1\}$.

  Let $n = k - b + 1$. For any $b_1, \hdots, b_{n} \in S_1$ such that $b_1 + \cdots + b_{n} \le |A|$. Consider any $T \subseteq A$ of size $b_1 + \hdots + b_{n}$. Partition $T = T_1 \cup T_2 \cup \cdots \cup T_{n}$ such that $|T_i| = b_i$ for all $i$. Consider the $k$-tuple
  \[(e_{T_1}, e_{T_2}, \hdots, e_{T_{n}}, e_{T}, \hdots, e_T),\]
  where $e_T$ appears $b-1 \ge 0$ times. We can verify that for each $i \in T$, there are exactly $b$ tuples {in this list} with $1$ in the $i$th coordinate. For any $i \not\in T$, there are $0$ tuples with $1$ in the $i$th coordinate. Thus,
  \[(f(e_{T_1}), \hdots, f(e_{T_n}), f(e_{T}), \hdots, f(e_{T})) \in Q.\]
  Since $f(e_{T_1}) = \cdots = f(e_{T_n}) = 1$, to avoid a contradiction, we must have that $f(e_T) = 0$, so $b_1 + \cdots + b_n \in S_0$.

  For any $a \in S_0$ and $b_1, \hdots, b_{n-1} \in S_1$ such that $a + b_1 + \cdots + b_{n-1} \le |A|$, consider $T \subseteq A$ of size $a + b_1 + \cdots + b_{n-1}$. Partition $T = T_0 \cup T_1 \cup \cdots \cup T_{n-1}$ such that $|T_0| = a$ and $|T_i| = b_i$ for all other $i$. It is easy to check that the following is a valid $k$-tuple
  \[
  (e_{T_1}, \hdots, e_{T_{n-1}}, \neg e_{T_0}, \hdots, \neg e_{T_0}, \neg e_T),
  \]
  where the are $k - n = b - 1$ copies of $\neg e_{T_0}$. Thus, since $f$ applies to the first $k - 1$ tuples is equal to $1$, $f(\neg e_T) = 0$, which implies by folding that $f(e_T) = 1$. Therefore, $a + b_1 + \cdots + b_{n-1} \in S_1$, as desired. 

  \textbf{End Cases}
  
  Thus, we have established that the conditions of Lemma \ref{lem:sum-game} hold for some $n(\Gamma)$. As stated at the beginning of the proof, we may apply the lemma to see that if $|A|$ grows arbitrarily large, then either $S_0 \cap S_1$ is nonempty for some $f$, which is an immediate contradiction, or $S_0$ contains all even integers between $d(n)$ and $|A|$ and $S_1$ contains all odd integers between $d(n)$ and $|A|$.
  To obtain a contradiction in this second case, {we show that $\Par_{L_1}$ is a minor of $f$.} {Observe that $f(e_{A}) = 1$, implying that $|A| \in S_1$, so $|A|$ is odd. Assume that $|A| \ge L_1(d(n)+1)$, then we can partition $A$ into $A_1 \cup \cdots \cup A_{L_1}$ such that for all $i \in [L_1]$, $|A_i|$ is odd and at least $d(n)$. Now consider, the map $\pi : [L] \to [L_1]$ defined as follows:}
\[
  {\pi(i) = \begin{cases}j & i \in A_j\\
      1 & i \not\in A.
    \end{cases}}
\]
{We claim that $f^{\pi} = \Par_{L_1}$. Since $f$ is folded, $f^{\pi}$ is also folded. Thus, it suffices to check that $f^{\pi}(x) = \Par_{L_1}$ when $x_1 = 0$. In that case,}
\begin{align*}
  {f^{\pi}(x)} &{= f(e_{\bigcup_{i : x_i = 1} A_i})}\\
                      &= {\begin{cases}0 & \sum_{i : x_i=1} |A_i| \text{ is even}\\
                          1 & \sum_{i : x_i = 1} |A_i| \text{ is odd} \end{cases}}\\
                      &{= \Par_{L_1}(x)\tag{{$|A_i|$ odd}}}.
\end{align*}
{The second line follows from the fact that each $A_i$ has size at least $d(n)$. Therefore, $\Par_{L_1} = f^{\pi} \in \Pol(\Gamma)$, a contradiction.}

Thus, $|A|$ is bounded, as desired.
\end{proof}

From this lemma, we can make an even stronger conclusion.

\begin{cor}\label{cor:noMalt-noPar}
  Let $\Gamma$ have the same properties as in Lemma \ref{lem:noMalt-noPar}. Let $f : \{0, 1\}^L \to \{0, 1\} \in \Pol(\Gamma)$ be any polymorphism and let $S_1, \hdots, S_{\ell}$ be disjoint subsets of $\{1, \hdots, L\}$ such that $f(e_{S_i}) = 1$ for all $i \in \{1, \hdots, \ell\}$. Then, $\ell \le c(\Gamma)$, where $c(\Gamma)$ is the same as in Lemma \ref{lem:noMalt-noPar}.
\end{cor}
\begin{proof}
Choose projection $\pi : \{1, \hdots, L\} \to \{1, \hdots, L\}$ such that for all $i \in \{1, \hdots, \ell\}$ and all $j \in S_i$, $\pi(j) = \min(S_i)$ and otherwise is the identity map. Consider $g = f ^ \pi$ which must also be a polymorphism of $\Gamma$. It is easy then to see that for all $i \in \{1, \hdots, \ell\}$, $g(e_{\min(S_i)}) = f(e_{S_i})$. Thus, $\ell \le c(\Gamma)$ by applying Lemma \ref{lem:noMalt-noPar} to $g$.
\end{proof}

\begin{lem}\label{lem:noMaj}
  Let $\Gamma$ be a symmetric, folded, idempotent family of promise relations such that\\$\Par_{L_1}, \Malt_{L_2}, \Maj_{L_3} \not\in \Pol(\Gamma)$ for some odd positive integers $L_1, L_2, L_3$. Then, there exists $C(\Gamma) \in \mathbb N$ such that for all $f \in \Pol(\Gamma)$, $f$ is $C(\Gamma)$-fixing.
\end{lem}
\begin{proof}
  Fix $f \in \Pol(\Gamma)$ of arity $f$. Pick a promise relation $(P, Q)$ of arity $k$ as guaranteed by Lemma \ref{lem:no-majority} such that $f \in \Pol(P, Q)$ for all $f \in \Pol(\Gamma)$.
  
  \textbf{Case 1,} $k \ge 3$ odd, $P = \Ham_k(\{(k+1)/2\}), Q = \Ham_k(\{0, \hdots, k - 1\})$.
  
  This case builds on techniques from Lemmas 4.2 and 5.4 of \cite{DBLP:journals/siamcomp/AustrinGH17}.

  Let $B \subseteq \{1, \hdots, L\}$ be the set of coordinates for which $f$ is \textit{somewhere-increasing}. That is, $B = \{i \in \{1, \hdots, L\} \: : \: \exists S \subseteq \{1, \hdots, L\}, f(e_{S\setminus i}) = 0, f(e_{S}) = 1\}$.

  {We claim that $f(e_S) = 1$, for all $S \supseteq B$.} {Otherwise if $f(e_S) = 0$, consider a sequence of subsets $S = S_0 \subset S_1 \subset \cdots \subset S_\ell = [L]$ such that $S_{i+1} \setminus S_i$ is always a singleton. Since $f(e_{[L]}) = 1$, there exists at least one $i$ such that $f(S_{i}) = 0$ and $f(S_{i+1}) = 1$, but the unique $j \in S_{i+1} \setminus S_i$ must be an element of $B$, contradiction.}

  Therefore, $f$ is $|B|$-fixing. Thus, if we deduce that $|B|$ is bounded by some $C$ for all $f$, then we know that all polymorphisms of $\Gamma$ are $C$-fixing.
  
  Let $a = (k-1)/2 \ge 1$. If $|B| < a$ then we are done. {Otherwise, w}e claim that for every subset $S \subseteq B$ of size $a$, we have that $f(e_S) = 1$. Let $S = \{i_1, \hdots, i_a\}$, and let $x^1, y^1, \hdots, x^a, y^a$ be witnesses for $i_1, \hdots, i_a \in B$. That is, $f(x^j) = 0, f(y^j) = 1$, $x^j_{i_j} = 0$, $y^j_{i_j} = 1$, and $x^j$ and $y^j$ are identical in all other coordinates. Consider now the $k$ tuples
  \[
  (\neg x^1, y^1, \hdots, \neg x^a, y^a, \neg e_S)
  \]
  It is easy to verify that in each coordinate $i \in \{1, \hdots, L\}$, exactly $a +1 = (k+1)/2$ of these tuples have their $i$th coordinate equal to $1$. {Thus, since} $f \in \Pol(P, Q)$, we have that not all of $f(\neg x^1), f(y^1), \hdots, f(\neq x^a), f(y^a), f(\neg e_S)$ are equal to $1$. Thus, since the first $2a$ are equal to $1$, we have that $f(\neg e_S) = 0$, so $f(e_S) = 1$, as desired.

  It is easy now to see that $|B| < (c(\Gamma)+1)a$, else we may construct disjoint $S_1, \hdots, S_{c(\Gamma)+1} \subseteq B$ of size equal to $a$, so $f(e_{S_1}), \hdots, f(e_{S_{c(\Gamma)+1}})$, violating Corollary \ref{cor:noMalt-noPar}. Thus, $|B|$ is bounded, so all $f$ are $C$-fixing for some $C(\Gamma)$ independent of $f$. 

  \textbf{Case 2,} $k \ge 3$, $P = \Ham_k(\{1, k\}), Q = \Ham_k(\{0, \hdots, k\} \setminus \{b\})$, $b \in \{2, \hdots, k-1\}$.

  Call $S \subseteq \{1, \hdots, L\}$ minimal if $f(e_S) = 1$ but $f(e_{S'}) = 0$ for all $S' \subset S$. We claim that if $S$ is minimal, then $|S| <  b$. Assume for contradiction that $S$ is minimal but $|S| \ge b$. Thus, we may find nonempty disjoint $S_1 \cup \cdots \cup S_{b} = S$. For each $i$, note that $f(e_{S\setminus S_i}) = 0$, so $f(e_{([L]\setminus S)\cup S_i}) = 1$ by folding. Furthermore, $f(e_{[L]\setminus S}) = 0$. Thus, consider the $k$-tuple
  \[
  (e_{([L]\setminus S)\cup S_1}, \hdots, e_{([L]\setminus S)\cup S_b}, e_{[L]\setminus S}, \hdots, e_{[L]\setminus S}). 
  \]
  where $e_{[L]\setminus S}$ appears $k - b$ times. It is easy to see that if $i \in S$, then the $i$th coordinate is equal to $1$ in exactly one element of this $k$-tuple, otherwise the $i$th coordinate is equal to $1$ in every $k$-tuple. Thus, the $i$th coordinates belong to $P$ for all $i \in [L]$. Since, $f \in \Pol(P, Q)$, we then have that
  \[
  (f(e_{([L]\setminus S)\cup S_1}), \hdots, f(e_{([L]\setminus S)\cup S_b}), f(e_{[L]\setminus S}), \hdots, f(e_{[L]\setminus S})) \in Q.
  \]
  But, the $k$-tuple has Hamming weight $b$, a contradiction. Thus, every minimal set has size strictly less than $b$. {Construct a sequence of subsets of $[L]$ as follows. Let $T_1$ be a minimum sized subset of $f$ such that $f(e_{T_1}) = 1$. Let $T_2$ be a minimum sized subset, if it exists of $f$ disjoint from $T_1$ such that $f(e_{T_2}) = 1$, and so forth. Note that each of this subsets is minimal. If we can construct $T_{c(|\Gamma|)+1}$, then $T_1, \hdots, T_{c(|Gamma|)+1}$ violate Corollary~\ref{cor:noMalt-noPar}. Otherwise, there is some $i < c(|\Gamma|)+1$ such that setting the coordinates of $T_1\cup \cdots \cup T_{i}$ to $0$ fixes $f$, and so $f$ is $c(|\Gamma|)(b-1)$-fixing. Either way we are done.}\footnote{{This proof idea originated in a subsequent paper~\cite{FicakKozikOlsakEtAl2019}}.}
\end{proof}

\subsection{Non-idempotent case}

Now, assume that our folded, symmetric family $\Gamma$ of promise relations has non-idempotent polymorphisms. If any polymorphism $f$ has the property that $f(0,\hdots, 0) = f(1,\hdots, 1),$ then folding is violated. Thus, $\Gamma$ is non-degenerate, so we may apply Lemma \ref{lem:idempotent} to yield that every polymorphism of $\Gamma$ is a polymorphism of the idempotent family $\Gamma'$ or it is the negation of a polymorphism of the idempotent family $\Gamma''$. Thus, if $\Gamma$ avoids Parity, Majority, \Maltsev{}, as well as their antis, then $\Gamma'$ and $\Gamma''$ both avoid Parity, Majority, and \Maltsev{}. By the previous section, the polymorphisms of $\Gamma'$ and $\Gamma''$ are $C$-fixing for some sufficiently large $C$. Since negating a folded polymorphism does not change that it is $C$-fixing, we have shown the following.

\begin{thm}\label{thm:junta}
Let $\Gamma$ be a finite, folded, symmetric family of promise relations. Assume there exist odd $L_1, \hdots, L_6$ such that $\Par_{L_1}$, $\Malt_{L_2}$, $\Maj_{L_3}$, $\overline{\Par}_{L_4}$, $\overline{\Malt}_{L_5},$  and $\overline{\Maj}_{L_6}$ are not polymorphisms of $\Gamma$. Then there exists $C(\Gamma)$ such that all polymorphisms of $\Gamma$ are $C$-fixing.
\end{thm}

\section{Hardness Arguments}\label{sec:hardness}

Now that we have a rather strong classification of polymorphisms
for folded, symmetric PCSPs, we are in a good position to interface it
with a reduction from Label Cover to actually demonstrate
$\mathsf{NP}$-hardness.

\begin{df}\label{df:lc}
An instance of Label Cover is based on a bipartite graph $G = (U, V, E)$. Each edge $e = (u, v)$ is associated with a projection $\pi_{e} : [R] \rightarrow [L]$ for some positive integers $R$ and $L$. A labeling is a pair of maps $\sigma_V : V \to [R]$, $\sigma_U : U \to [L]$. A labeling \textit{satisfies} the instance if for all $(u, v) \in E$, $\pi_{(u, v)}(\sigma_V(v)) = \sigma_U(u)$.
\end{df}

The PCP theorem combined with parallel repetition gives the following well-known hardness of Label Cover which is the starting point for most inapproximability results.

\begin{prop}
\label{prop:lc}
For any $\eta >  0$, given an instance of Label Cover it is $\mathsf{NP}$-hard to distinguish between the two cases:
\begin{itemize}
\item Completeness: There exists a labeling $\sigma_V, \sigma_U$ that satisfies every edge.
\item Soundness: No labeling $\sigma_V, \sigma_U$ can satisfy a fraction $\eta$ of the edges.
\end{itemize}
\end{prop}

\begin{thm}\label{thm:hardness}
Let $\Gamma$ be a folded, finite family of promise relations. Suppose that there exists a universal constant $C=C(\Gamma) < \infty$ such that every polymorphism of $\Gamma$ is $C(\Gamma)$-fixing. Then $\PCSP(\Gamma)$ is $\mathsf{NP}$-hard.
\end{thm}
\begin{proof}
  The proof is via reduction from the hardness of Label Cover as stated in Proposition~\ref{prop:lc}, for the parameter $\eta = 1/C^2$.  The proof is a simplification of the proof of Theorem 1.1 of \cite{DBLP:journals/siamcomp/AustrinGH17}.

  Let $G = (U, V, E)$ be our instance with maps $\pi_e : [R] \to [L]$. As noted in Remark 4.7 of \cite{DBLP:journals/siamcomp/AustrinGH17}, $L$ and $R$ are functions of $\eta$ and thus are independent of the size of $G$. We now create a $\Gamma$-$\PCSP$ $\Psi = (\Psi_P, \Psi_Q)$. For each $u \in U$, identify the vertex with $2^L$ variables which we denote by $f_u(x)$ where $x \in \{0, 1\}^L$ and $f_u : \{0, 1\}^L \to \{0, 1\}$. For all $(P, Q) \in \Gamma$ and $x^1, \hdots, x^L \in P$ (possibly with repetition) we enforce the constraint
  \[
  P(f_u(x^1_1, \hdots, x^L_1), \hdots, f_u(x^1_k, \hdots, x^L_k))
  \]
  in $\Psi_P$, with the corresponding constraint in $\Psi_Q$. From the perspective of $\Psi_Q$, $f_u$ is a polymorphism of $\Gamma$. Likewise, for each $v \in V$, identify $2^R$ variables which we denote by $f_v(y)$ where $y \in \{0, 1\}^R$ and $f_v : \{0, 1\}^R \to \{0, 1\}$. Again, using the constraints of $\Gamma$, we may specify that $f_v$ is a polymorphism from the perspective of $\Psi_Q$.

  Next, we specify the edge constraints, which we do in a manner greatly simplifying that of \cite{DBLP:journals/siamcomp/AustrinGH17}. For each $e = (u, v) \in E$ and for any $x \in \{0, 1\}^L$ and $y \in \{0, 1\}^R$ such that $x_{\pi_e(i)} = y_i$ for all $i \in [R]$, we specify that $f_u(x) = f_v(y)$. Note that $\Gamma$ might not have an equality constraint, but we can implicitly introduce one by using the same variable for $f_u(x)$ and $f_v(y)$ when constructing $\Psi$. For a specific $(u, v) \in E$, these constraints maintain that $f_v^{\pi_{(u, v)}} = f_u$ (in both $\Psi_P$ and $\Psi_Q$).

  To show that this is a valid reduction, we need to show that both completeness and soundness hold (see Lemmas 4.5 and 4.6 of \cite{DBLP:journals/siamcomp/AustrinGH17}).

  \begin{itemize}
  \item Completeness: If there exists a labeling $\sigma_U, \sigma_V$ satisfying every edge of the Label Cover instance, let $f_u(x) = x_{\sigma_U(u)}$ and $f_v(x) = x_{\sigma_V(v)}$. These satisfy the constraints for $\Psi_P$ since dictators are polymorphisms of $(P, P)$ (as well as $(P, Q)$) for all $(P, Q) \in \Gamma$. The equal constraints are also satisfied since if $x_{\pi_e(i)} = y_i$ for some $e = (u, v) \in E$, then $f_u(x) = x_{\sigma_U(u)} = y_{\sigma_V(v)} =  f_v(y)$, as desired. Thus, $\Psi_P$ is satisfiable when our Label Cover instance is satisfiable.

  \item Soundness: Assume for {the} sake of contradiction, that a satisfying assignment to $\Psi_Q$ exists. For each $u \in U$, $v \in V$, $f_u$ and $f_v$ are $C$-fixing. Thus, we may define $S_u \subseteq [L], S_v \subseteq [R]$ such that $f_u(x) = f_u(1, \hdots, 1)$ and $f_v(y) = f_v(1, \hdots, 1)$ if $x_i = 1$ for $i \in S_u$ and $y_j = 1$ for $j \in S_v$. Since $S_u$ is $C$-fixing, we can let $|S_u| \le C$. {By the same logic, $|S_v| \le C$.}

    We claim that for every edge $e= (u, v) \in E$, $S_u \cap \pi_e[S_v]$ is nonempty (where $\pi_e[S_v] = \{\pi(s) : s \in S_v\}$). By virtue of the equality constraints, $f_u = f_v^{\pi_e}$; thus we have that $f_u(e_{\pi[S_v]}) = f_v(e_{S_v}) = f_v(1, \hdots, 1) = f_u(1, \hdots, 1).$ Thus, as $\Gamma$ is folded, $f_u$ must be folded, so $f_u(e_{[L] \setminus \pi_e[S_v]}) = \neg f_u(1, \hdots, 1)$. If $S_u$ and $\pi_e[S_v]$ were disjoint, then $S_u \subseteq [L]\setminus \pi_e[S_v]$, so by the definition of $S_u$, $f_u(e_{[L] \setminus \pi_e[S_v]}) = f_u(1, \hdots, 1)$, contradiction. Thus, $S_u$ and $\pi_e[S_v]$ intersect non-trivially.

    Due to this fact, we can show a $\eta$-approximate labeling exists for our label cover instance as in \cite{DBLP:journals/siamcomp/AustrinGH17} and typical for Label Cover reductions. For each $u \in U$, select $\sigma_U(u)$ uniformly at random from $S_u$. Likewise, for each $v \in V$, select $\sigma_V(v)$ uniformly at random from $S_v$. Since for any given $e = (u, v) \in E$ we have that $S_u$ and $\pi_e[S_v]$ have a common intersection and both sets have size at most $C$, $\sigma_U(u) = \pi_e(\sigma_V(v))$ with probability at least {$\eta \geq 1/C^2$}. Thus, the expected number of constraints satisfied by a random labeling is at least $\eta$. Hence, there exists a labeling which satisfies at least $\eta$-fraction of the constraints, as desired.
  \end{itemize}

  Thus, we have completed our reduction, so $\PCSP(\Gamma)$ is $\mathsf{NP}$-hard.
\end{proof}

Hence, we have completed the proof of Theorem \ref{thm:main-result} by combining Theorems \ref{thm:alg-result}, \ref{thm:junta}, and \ref{thm:hardness}.

\section{General Theory of Promise CSPs}\label{app:theory}

This section contains a number of additional observations about Promise CSPs which help to build a broader theory of Promise CSPs.

\subsection{Galois Correspondence of Polymorphisms}\label{app:Galois}

In this part, we show that for any finite family $\Gamma$ of
promise relations of any finite arity, we show that $\Pol(\Gamma)$
captures the computational complexity of $\PCSP(\Gamma)$ {in} the
following precise sense. To do this, we show a \emph{Galois correspondence} between families of polymorphisms and families of promise relations that are closed under a form of reduction.

\begin{thm}\label{thm:Galois}
  Let $\Gamma$ and $\Gamma'$ be families of promise relations such that $\Pol(\Gamma) \subseteq \Pol(\Gamma')$. Then, there is a polynomial-time reduction from $\PCSP(\Gamma')$ to $\PCSP(\Gamma)$.
\end{thm}

This result is the promise-analogue to Theorem 3.16 of
\cite{Chen:2009}, originally established by \cite{jeavons88}, which
holds for traditional CSPs. Our proof has similar structure to that of
\cite{Chen:2009}. Pippenger~\cite{Pippenger2002} in Section 2 of his
paper proves a variation of the Galois correspondence between promise
relations and their polymorphisms, although not in
this particular complexity-theoretic formulation.

In fact, we show the polynomial-time reduction is of a very local form. Let $\EQUAL = \{(i, i) \: : \: \in D\}$ be the relation which specifies that two variables are equal. Since we have been allowing repetition of variables, this relation has been essentially implicit.

\begin{df}\label{df:ppp-definable}
  Let $\Gamma$ be a finite family of promise relations. We say that a promise relation $(P', Q') \in D^k \times D^k$ is \textit{positive primitive promise definable} (shortened to \textit{ppp-definable}) from $\Gamma$ if there exists a $\Gamma\cup\{\EQUAL\}$-$\PCSP$ $\Psi = (\Psi_P, \Psi_Q)$ on $k + \ell$ variables such that
  \begin{itemize}
  \item For all $(x_1, \hdots, x_k)\in P'$, there exists $(y_1, \hdots, y_{\ell})$ such that $(x_1, \hdots, x_k, y_1, \hdots, y_{\ell})$ is a satisfying assignment to $\Psi_P$.
  \item For all satisfying assignments $(z_1, \hdots, z_{k + \ell})$ to $\Psi_Q$, $(z_1, \hdots, z_k) \in Q'$.
  \end{itemize}
  We say that a finite family of promise relations $\Gamma'$ is ppp-definable from $\Gamma$ if every $(P', Q') \in \Gamma'$ is ppp-definable from $\Gamma$.
\end{df}

In particular, note that if $(P, Q)$ and $(P', Q')$ have the same arity and $P' \subseteq P \subseteq Q \subseteq Q'$ then $(P', Q')$ is ppp-definable from $(P, Q)$ by letting $(\Psi_P, \Psi_Q) = (P, Q)$. We also note that ppp-definability is reflexive ($\Gamma$ is ppp-definable from $\Gamma$) and transitive: if $\Gamma'$ is ppp-definable from $\Gamma$ and $\Gamma''$ is ppp-definable from $\Gamma'$ then $\Gamma''$ is ppp-definable from $\Gamma$. We have that ppp-definability is a formalization {of} the notion of a \emph{gadget reduction} in \cite{DBLP:journals/siamcomp/AustrinGH17} (see Proposition 3.1).

Our notion of ppp-definability is a direct generalization of {the} notion of \textit{pp-definability} for normal CSP relations defined in \cite{Chen:2009}. If $\Gamma'$ is ppp-definable from $\Gamma$, there is a corresponding polynomial-time reduction from $\PCSP(\Gamma')$ to $\PCSP(\Gamma)$ by replacing each $(P', Q') \in \Gamma'$ clause with a corresponding $(\Psi_P, \Psi_Q)$ clause (adding in any auxiliary variables), which can can be implemented with clauses from $\Gamma$ and yields only a constant-factor blowup. It is straightforward to verify that this reduction is valid. As noted in \cite{Chen:2009}, this reduction can be done in logarithmic space.

In establishing the Galois correspondence, one important ppp-definition from $\Gamma$ {is the} promise relation of polymorphisms of $\Gamma$. {This definition is instrumental in showing that every relaxation of PCSP is also a ppp-definition and vice versa.}

\begin{prop}\label{prop:ppp-poly}
  Let $L$ be a positive integer. The following promise relation $S_L \subseteq T_L \subseteq D^{D^L}$ is ppp-definable from $\Gamma$:
  \begin{align*}
  S_L &= \{f : D^L \to D \: :\: f \in \Pol(P, P)\text{ for all }(P, Q) \in \Gamma\}\\
  T_L &= \{f : D^L \to D\: :\: f \in \Pol(P, Q)\text{ for all }(P, Q) \in \Gamma\},
  \end{align*}
  where we identify a function $f \in D^L \to D$ as a vector of $|D|^L$ variables.
\end{prop}
\begin{proof}
  Using the definition of a polymorphism, one can specify that $f$ is a polymorphism of $\Pol(P, Q)$ of specific arity in terms of a fixed number of $Q$-clauses. Replacing those $Q$-clauses with $P$-clauses exactly characterizes that $f \in \Pol(P, P)$.
\end{proof}

With these facts established, we may now prove the theorem. The proof is quite similar to and was inspired by the second half of Theorem 3.13 of \cite{Chen:2009}.

\begin{proof}[Proof of Theorem \ref{thm:Galois}]
  It suffices to show that every promise relation $(P', Q') \in \Gamma'$ is ppp-definable from $\Gamma$. Let $k$ be the arity of $(P', Q')$ and let $m = |P'|$. Let $x^1, \hdots, x^m$ be some ordering of the elements of $P'$. Define $y^1, \hdots, y^k \in D^m$ such that $y^i_j = x^j_i$ for all $i \in [k], j \in [m]$. Now from Proposition \ref{prop:ppp-poly}, we have that $(S_m, T_m)$ is ppp-definable from $\Gamma$. Now, consider the following promise relation $(S'_m, T'_m)$ of arity $k$.
  \begin{align*}
    S'_m &= \{(f(y^1), \hdots, f(y^k))\: : \: f \in S_m\}\\
    T'_m &= \{(f(y^1), \hdots, f(y^k))\: : \: f \in T_m\}.
  \end{align*} 
  We have that $(S'_m, T'_m)$ is ppp-definable from $(S_m, T_m)$ since every $x \in S'_m$ can be built up into a corresponding element of $S_m$ and every $y \in T_m$ can be stripped down to an element of $T'_m$. Note that this is the case even if $y^i = y^j$ for some distinct $i, j \in [k]$ by using the $\EQUAL$ relation.

  We claim that $P' \subseteq S'_m \subseteq T'_m \subseteq Q'$. First, for all $i \in [m]$, consider the unique projection map $\pi_i : D^m \to D$ given by $\pi_i(y) = y_i$. Clearly $\pi_i \in S_m$. Thus, $(\pi_i(y^1), \hdots, \pi_i(y^k)) = (y^1_i, \hdots, y^k_i) = x^i \in S'_m$. Thus, $P' \subseteq S'_m$. Second, we can see that $S'_m \subseteq T'_m$ since $S_m \subseteq T_m$. Third, note that $T'_m \subseteq {T_m}(P') \subseteq {\Pol(\Gamma)}(P')$. Since $\Pol(P', Q') \supseteq \Pol(\Gamma)$, we have that $Q' \supseteq {\Pol(\Gamma)}(P')$. Thus, $T'_m \subseteq Q'$.

  Thus, therefore $(P', Q')$ is ppp-definable from $(S'_m, T'_m)$. By transitivity, we have that $(P', Q')$ is ppp-definable from $\Gamma$, so $\Gamma'$ is ppp-definable from $\Gamma$.
\end{proof}

\subsection{Polymorphism-only description of PCSPs}

In this section, we establish a
necessary and sufficient set of conditions on a set $\mathcal F$ of
functions over domain $D$ for which there exists some finite $\Gamma$ such
that $\mathcal F = \Pol(\Gamma)$. Pippenger~\cite{Pippenger2002}
proved such a characterization in the case that $\Gamma$ may or may
not \emph{have infinitely many relations} using similar ideas. Recall the definition of a projection of a polymorphism.

\begin{df}
  Let $f : D^L \to D$ be a function, and let $\pi : [L] \to [R]$ be any map which we call a \emph{projection}. The projection of $f$ with respect to $\pi$ is the function $f^{\pi} : D^R \to D$ such that for all $y \in D^R$, $f^{\pi}(y) = f(x)$, where $x \in D^L$ is the unique $L$-tuple such that
  \[
  x_i = y_{\pi(i)}, \text{ for all $i \in [L]$.}
  \]
\end{df}

Note that in a projection it might be the case that $R \ge L$. We say that a family $\mathcal F$ of functions over domain $D$ is \emph{projection-closed} if for all $L, R \in \mathbb N$, all $f \in \mathcal F$ of arity $L$, and all maps $\pi : [L] \to [R]$, $f^{\pi} \in \mathcal F$.

Another technical property we require of $\mathcal F$ is that it is \emph{finitizable}. This means there exists some finite arity $R \in \mathbb N$, called the \emph{finitized arity} such that $f : D^L \to D$ is an element of $\mathcal F$ if and only if for all $\pi : [L] \to [R]$, the projection $f^{\pi}$ is an element of $\mathcal F$. Intuitively, this finitization property says that some finite arity of $\mathcal F$ captures all of the meaningful information about what is contained in $\mathcal F$. This is directly analogous to the property that our set $\Gamma$ of promise relations is finite.

Surprisingly, these two properties--that $\mathcal F$ is projection-closed and finitizable--perfectly capture the families of the form $\Pol(\Gamma)$ for some $\Gamma$ as long as we stipulate that $\mathcal F$ contains the \emph{identity function}: $\id_D : D \to D$ such that ${\id}_D(x) = x$ for all $x \in D$.\footnote{If we broaden our definition of PCSPs (as mentioned in the introduction) so that instead of $P \subseteq Q$, there is some unary map $\sigma : D_1 \to D_2$ such that $\sigma(P) \subseteq Q$, then the condition $\id_D \in \mathcal F$ can be replaced with $\sigma \in \mathcal F$ for some unary function $\sigma$.}

\begin{thm}\label{lem:finitization}
  Let $\mathcal F$ be a family of functions over domain $D$. Then, there exists a finite family $\Gamma$ of promise relations such that $\mathcal F = \Pol(\Gamma)$ if and only if $\mathcal F$ is both projection-closed and finitizable and $\id_D \in \mathcal F$.
\end{thm}

We start by showing that these two properties are necessary.

\begin{claim}\label{claim:finitization-necessary}
  Let $\Gamma = \{(P_i, Q_i) : P_i \subseteq Q_i \subseteq D^{k_i}\}$ be a finite family of promise relations with domain $D$. Then, $\Pol(\Gamma)$ is both projection-closed and finitizable.
\end{claim}

\begin{proof}
  \emph{projection-closed:} Let $f : D^L \to D$ be a polymorphism of $\Gamma$ and let $\pi : [L] \to [R]$ be a map. We claim that $f^{\pi} : D^R \to D$ is also a polymorphism of $\Gamma$. Consider all $(P_i, Q_i)$ and $y^{(1)}, \hdots, y^{(R)} \in P_i$. We need to show that $f^{\pi}(y^{(1)}, \hdots, y^{(R)}) \in Q_i$. Consider $x^{(1)}, \hdots, x^{(L)} \in P_i$ such that $x^{(j)} = y^{(\pi(j))}$ for all $j \in [L]$. From the definition of $f^{\pi}$ it is then easy to see that
  \[
  f^{\pi}(y^{(1)}, \hdots, y^{(R)}) = f(x^{(1)}, \hdots, x^{(L)}) \in Q_i,
  \]
  as desired.

  \emph{finitizable:} Let $R = \max_{(P_i, Q_i) \in \Gamma} |P_i|$. Crucially, this maximum exists since $\Gamma$ is finite. Since $\Pol(\Gamma)$ is projection-closed, for all $f \in \Pol(\Gamma)$ of arity $L$ and all $\pi : [L] \to [R]$, we have that $f^{\pi} \in \Pol(\Gamma)$.

  Now, consider any $f \not\in \Pol(\Gamma)$ of arity $L$, we would like to show that there exists $\pi : [L] \to [R]$ such that $f^{\pi} \not\in \Pol(\Gamma)$. Since $f \not\in \Pol(\Gamma)$, there exists $(P_i, Q_i) \in \Gamma$ and $x^{(1)}, \hdots, x^{(L)} \in P_i$ such that $f(x^{(1)}, \hdots, x^{(L)}) \not\in Q_i$. Since $R \ge |P_i|$, there exists an injective map $\sigma : P_i \to [R]$. Let $\pi : [L] \to [R]$ be $\pi(i) = \sigma(x^{(i)}).$ By nature of $\pi$, we can select $y^{(1)}, \hdots, y^{(R)} \in P_i$ such that $y^{(r)} = x^{(\pi^{-1}(r))}$ for all $r \in \Im(\pi)$ and $y^{(r)} = 1$ otherwise. (If $r \in R$ is not in the image of $\pi$, then we may make an arbitrary choice.) Note that if $\pi(j_1) = \pi(j_2)$ then $x^{(j_1)} = x^{(j_2)}$ so this choice of $y^{(j)}$'s is well-defined. From the definition of a projection,
  \[
  f^{\pi}(y^{(1)}, \hdots, y^{(R)}) = f(x^{(1)}, \hdots, x^{(L)}) \not\in Q_i,
  \]
  as desired. Therefore $\Pol(\Gamma)$ is finitizable.
\end{proof}

Note that since we stipulate that $P \subseteq Q$ for all $(P, Q) \in \Gamma$, we immediately have that $\id_D$ is a polymorphism of $\Gamma$. Much more difficultly, we show that these two properties are sufficient.

\begin{lem}\label{lem:finitization-sufficient}
  Let $\mathcal F$ be a domain-$D$ family of functions which is both projection-closed and finitizable as well as has $\id_D$ as an element. Then, there exists a family $\Gamma$ of finitely many promise relations such that $\Pol(\Gamma) = \mathcal F$.
\end{lem}

\begin{proof}
  Let $R \in \mathbb N$ be the finitized arity of $\mathcal F$. Identify the integers of $[|D|^R]$ with elements of $D^R$. Our choice of $\Gamma$ will consist of a single promise relation $P \subseteq Q \subseteq D^{|D|^R}$, where each $f \in D^{[|D|^R]}$ will be identified with a function $f : D^R \to D$ in the canonical way. We let $f \in P$ if and only if there exists $j \in [R]$ such that $f(x) = x_j$ for all $x \in D^R$. We let $Q = \{f \in \mathcal F \mid \text{$f$ has arity $R$}\}$. Since $\mathcal F$ has the identity function and is projection-closed, we have that $P \subseteq Q$. Thus, $\Gamma$ is a finite promise relation.

  Now that we have constructed $\Gamma$, we need to show $\Pol(\Gamma) = \mathcal F$. Enumerate the elements of $P$ as $y^{(1)}, \hdots, y^{(R)}$, where $y^{(j)}(x) = x_j$ for all $j \in [R]$ and $x \in D^R$. With this enumeration, we have the property that for all $g \in \Pol(P, Q)$ of arity $R$, $g(y^{(1)}, \hdots, y^{(R)}) = g$. This is because for all $x\in D^R$,
  \[g(y^{(1)}(x), \hdots, y^{(R)}(x)) = g(x_1, \hdots, x_R) = g(x).\]
  (Thus the $y^{(i)}$'s are like a long code test.)

  First, we show that $\Pol(P, Q) \subseteq \mathcal F$. Consider any $f \in \Pol(P, Q)$ of arity $L$ as well as any $\pi : [L] \to [R]$. Pick $x^{(1)}, \hdots, x^{(L)} \in P$ such that $x^{(j)} = y^{(\pi(j))}$ for all $j \in [L]$. Thus,
  \[
  Q \ni f(x^{(1)}, \hdots, x^{(L)}) = f^{\pi}(y^{(1)}, \hdots, y^{(R)}) = f^{\pi}.
  \]
  Thus, $f^{\pi} \in Q \subseteq \mathcal F$ for all $\pi : [L] \to [R]$. Thus, $f \in \mathcal F$ since $\mathcal F$ is finitizable, as desired.

  Last, we show that $\mathcal F \subseteq \Pol(P, Q)$. For every $f \in \mathcal F$, we need to show that for all $x^{(1)}, \hdots, x^{(L)} \in P$, we have that $f(x^{(1)}, \hdots, x^{(L)}) \in Q$. Since $y^{(1)}, \hdots, y^{(R)}$ is an enumeration of the elements of $P$, there is a unique $\pi : [L] \to [R]$ such that $x^{(j)} = y^{(\pi(j))}$ for all $j \in [L]$. Then, we have that
  \[
  f(x^{(1)}, \hdots, x^{(L)}) = f^{\pi}(y^{(1)}, \hdots, y^{(R)}) = f^{\pi}.
  \]
  Since $\mathcal F$ is projection closed, $f^{\pi} \subseteq \mathcal F$. Therefore, $f^{\pi} \in Q$ because $f^{\pi}$ has arity $R$. Thus, $f(x^{(1)}, \hdots, x^{(L)}) \in Q$, as desired.

  Hence, $\mathcal F = \Pol(P, Q)$.
\end{proof}

Claim \ref{claim:finitization-necessary} and Lemma \ref{lem:finitization-sufficient} together establish Theorem~\ref{lem:finitization}.

\subsection{Analogous characterization for CSPs}

We now extend Theorem~\ref{lem:finitization} to show that the same characterization holds for CSPs as long as we add the condition that our set of functions form a \emph{clone} (defined below). For our purposes, a CSP is a PCSP $\Gamma$ in which $P = Q$ for all $(P, Q) \in \Gamma$.

As known in the CSP literature {(}e.g., \cite{Chen:2009}{)}, the family of polymorphisms of a CSP $\Lambda$ have the additional property {it} is a \emph{clone}.\footnote{{As pointed out to the authors after writing this, the collection of clones which correspond to $\CSP(\Gamma)$ for finite $\Gamma$ are ``finitely related clones'' (see, e.g.,~\cite{markovic2012finitely}.)}}That is, for all $f \in \Pol(\Lambda)$ of arity $L_1$, and all $g_1, \hdots, g_{L_1} \in \Pol(\Lambda)$ of arity $L_2$, we have that $h(x^{(1)}, \hdots, x^{(L_1)}) {:=} f(g_1(x^{(1)}), \hdots, g_{L_1}(x^{(L_1}{))}$ is a polymorphism of $\Lambda$ of arity $L_1L_2$. It turns out this property is necessary and sufficient for characterizing CSPs from their polymorphisms.

\begin{lem}\label{lem:CSP-finitization}
  Let $\mathcal F$ be a family of functions over the domain $D$. Then, there exists a CSP $\Lambda$ such that $\mathcal F = \Pol(\Lambda)$ if and only if $\mathcal F$ {is} finitizable, a clone, and contains the identity.
\end{lem}

\begin{proof}
  As stated previously, $\Pol(\Lambda)$ {is} finitizable, a clone, and contains the identity. Thus, it suffices to show the converse.

  Assume that $\mathcal F$ finitizes at arity $R$. As shown in Lemma \ref{lem:finitization-sufficient}, $\mathcal F = \Pol(P, Q)$, where $P \subseteq Q \subseteq D^{D^R}$. In this case, $P$ are the $R$ projection functions from $D^R$ to $D$ and $Q$ is the set of arity-$R$ functions of $\mathcal F$. Since, we now have that $\mathcal F$ is a clone, we claim that $\mathcal F = \Pol(Q, Q)$.

  First, we have that $\Pol(Q, Q) \subseteq \Pol(P, Q) = \mathcal F$ since $P \subseteq Q$, so membership in $\Pol(Q, Q)$ is a more strict condition. To show the reverse inclusion $\mathcal F \subseteq \Pol(Q, Q)$, consider any $f \in \mathcal F$ of arity $L$. We need to show for all $g_1, \hdots, g_L \in Q$, we have that $f(g_1, \hdots, g_L) \in Q$. Since $\mathcal F$ is a clone, we immediately have that $f(g_1, \hdots, g_L) \in \mathcal F$. Furthermore, $f(g_1, \hdots, g_L)$ has arity $R$ so $f(g_1, \hdots, g_L) \in Q$.
\end{proof}

\subsection{Significance toward establishing complexity of PCSPs}

{T}hese results liberate us from ever thinking about $\Gamma$, and instead we can think entirely in terms of establishing the easiness/hardness of projection-closed, finitized families of functions. {As discussed in the introduction, subsequent work (e.g., \cite{BartoBulinKrokhinEtAl2019}) has shown that it sufficies to only consider the \emph{identities} which the polymorphisms satisfy.}
Another point of consideration is the case in which there are \textit{infinitely} many relations in our (P)CSPs (although keeping a finite domain). As a computational problem, one can define the (non-uniform) computational complexity of a PCSP $\Gamma$, in the style of the compactness theorem, to be the supremum of the computational complexities of all finite subsets $\Gamma' \subseteq \Gamma$. Another common (uniform) definition is that the relations used in any particular CSP are encoded as part of the input (using some canonical encoding). The \emph{local-global conjecture} (e.g., \cite{BodirskyGrohe2008}) states that these two notions of complexity should be identical for infinite case. Such a conjecture could also be made for PCSPs, although we doubt the veracity of such a claim for the following reason. Once we allow infinitely many relations into our PCSPs, the possible characterizations of polymorphisms expand to all projection-closed families (that is, the finitization condition can be dropped). As a result, it seems quite tempting that an $\mathsf{NP}$-intermediate PCSP could be constructed by adapting the techniques used to prove Ladner's theorem \cite{Ladner1975}.

\subsection{Every PCSP is Equivalent to a Promise Digraph Homomorphism.}\label{app:digraph}

We show that for any finite $\PCSP(\Gamma)$ over a
domain $D$, there exist directed graphs $H_{\Gamma}, H'_{\Gamma}$
with a homomorphism from $H_{\Gamma}$ to $H'_{\Gamma}$ such that
$\PCSP(\Gamma)$ is polynomial time equivalent to the promise digraph
homomorphism problem on $(H_{\Gamma}, H'_{\Gamma}).$ Below is a formal
definition of the promise digraph problem.

\begin{df} 
  Let $H, H'$ be a pair of directed graphs such there is a homomorphism
  $\phi : H \to H'$. The promise decision problem $\mathsf{PDGH}(H,
  H')$ (promise directed graph homomorphism) is given a directed graph
  $G$ distinguish between the two possible cases.
  \begin{itemize}
    \item[YES.] There is a homomorphism from $G$ to $H$.
    \item[NO.] There is no homomorphism from $G$ to $H'$.
  \end{itemize}
\end{df}

\begin{thm}
  Let $D$ be a finite domain, and let $\Gamma$ be any promise relation
  over $D$. Then, there exists a pair of directed graphs $H_{\Gamma}$
  and $H'_{\Gamma}$ such that $H_{\Gamma} \subseteq H'_{\Gamma}$ and
  $\PCSP(\Gamma)$ is polynomial-time equivalent to
  $\mathsf{PDGH}(H_{\Gamma}, H'_{\Gamma})$.
\end{thm}

Our proof is an adaptation of corresponding result for
CSPs by Feder and Vardi, see Theorem 11 of \cite{feder-vardi}. We cite
the following slightly stronger claim which is implied in the proof of
Theorem 11. We let $\Hom(H)$ denote the digraph graph homomorphism
problem of determining whether there exists a homomorphism from an
input graph $G$ to $H$. To do this, we need to define a notion of an
``oblivious reduction.''

\begin{df}
  Let $\Gamma = \{R_i\}$ be a CSP over a domain $D$. Let $\Psi$ be an
  instance of $\CSP(\Gamma)$, the \emph{underlying colored directed hypergraph}
  of $\Psi$ is the hypergraph whose vertices are the
  variables of $\Psi$ and whose (directed) hyperedges correspond to the tuples on which 
the relations of $\Gamma$ are applied, with the hyperedge having color $i$ when $R_i$ is applied.
\end{df}

\begin{df}
\label{def:oblivious-FV}
  Let $\Gamma$ be a CSP over a domain $D$ and let $H$ be a directed
  graph dependent on $D$.

  A reduction $\Psi \mapsto G_{\Psi}$ from $\CSP(\Gamma)$ to $\Hom(H)$ is \emph{oblivious} if
 $G_{\Psi} = \sigma(D,X)$ for some function $\sigma$ independent of $\Gamma$ applied to 
 the domain $D$ and underlying
  colored directed hypergraph $X$ of $\Psi$.

  Conversely, a reduction $G \mapsto \Psi_G$ from $\Hom(H)$ to $\CSP(\Gamma)$ is
  \emph{oblivious} if $\Psi_G = \eta(D,G)$ for some function $\eta$ independent of $H$.
  
\end{df}

\begin{thm}[Feder and Vardi, \cite{feder-vardi}]\label{thm:fed-vard}
  Let $\Gamma$ be a CSP over a domain $D$ such that $\Gamma = (R, S)$,
  where $R$ has arity $1$ and $S$ has arity $2$ with the property that
  the projections of the first and second coordinate of $S$ are both
  all of $D$. Then, there is a digraph $H_{\Gamma}$ such that
  $\CSP(\Gamma)$ is equivalent to the graph homomorphism problem on
  $H_{\Gamma}$.  Furthermore, the reductions mapping instances of
  $\CSP(\Gamma)$ to those of $\Hom(H_{\Gamma})$ and the reductions
  mapping instances of $\Hom(H_{\Gamma})$ to $\CSP(\Gamma)$ are
  both oblivious in the sense of Definition~\ref{def:oblivious-FV}.
\end{thm}

 Below we sketch the
details necessary to extend this argument to PCSPs. 

\begin{thm}
  Let $\Gamma$ be a PCSP, then there exists a pair of digraphs $H$ and
  $H'$ such there is a homomorphism from $H$ to $H'$ and
  $\PCSP(\Gamma)$ is polynomial time equivalent to $\mathsf{PDGH}(H,
  H')$.
\end{thm}

\begin{proof}
The proof technique borrows significantly from \cite{feder-vardi}.

Assume that our PCSP is $\Gamma = \{(P_i, Q_i) {: i \in [r]}\}$. Let $P' = \prod_{i}
P_i$ and $Q' = \prod_i Q_i$. Let $\Gamma' = \{(P', Q')\}$. It is easy to
check that $\PCSP(\Gamma)$ is polynomial time equivalent to $\PCSP(\Gamma')$.

Assume the single relation of $\Gamma'$ has arity $k$. Now, view
$\Gamma'$ as a CSP over the domain $D' = D^k$. Note that
$\PCSP(\Gamma')$ over the domain $D'$ is not necessarily equivalent to
$\PCSP(\Gamma')$ over the domain $D$, since we lose the ability to
specify the same variable in different coordinates. To mitigate this,
we add an additional ``shift operator'' $S = \{(x, y) : x,y \in D^k;
x_i = y_{i+1}, i \in \{1, \hdots, k-1\}\}.$ Then, we have that
$\PCSP(\Gamma'\cup \{(S, S)\})$ over domain $D'$ is polynomial-time
equivalent $\PCSP(\Gamma')$ over domain $D$.  See \cite{feder-vardi}
for more details.

Then, applying Theorem~\ref{thm:fed-vard}, to $\{P', S\}$ and $\{Q',
S\}$ over domain $D'$, we obtain digraphs $H_{P'}$ and $H_{Q'}$ such that $\CSP(\{P',
S\})$ is polynomial-time equivalent to $\Hom(H_{P'})$ and $\CSP(\{Q', S\})$
is polynomial-time equivalent to $\Hom(H_{Q'})$. From the nature of
the construction of $H_{P'}$ and $H_{Q'}$ in \cite{feder-vardi}, we
know that $H_{P'} \subseteq H_{Q'}$, which trivially implies there is
homomorphism between them.

Since the reductions to $\Hom(H_{P'})$ and $\Hom(H_{Q'})$ are
oblivious, any instance $\Psi$ of
$\PCSP(\Gamma'\cup \{(S, S)\})$ reduces to a digraph $G_{\Psi}$ such
that there is a homomorphism from $G_{\Psi}$ to $H_{P'}$ if and only
if $\Psi_P$ is satisfiable. Likewise, there is a homomorphism from
$G_{\Psi}$ to $H_{Q'}$ if and only if $\Psi_Q$ is satisfiable. Thus,
any instance $\Psi$ of $\PCSP(\Gamma'\cup \{(S, S)\})$ reduces to an
instance $G_{\Psi}$ of $\mathsf{PDGH}(H_{P'}, H_{Q'})$.

Similarly, the reductions from $\Hom(H_{P'})$ and $\Hom(H_{Q'})$ are
oblivious, so any instance $G$ of
$\mathsf{PDGH}(H_{P'}, H_{Q'})$ reduces to an instance $\Psi_G$ of
$\PCSP(\Gamma'\cup \{(S, S)\})$.
\end{proof}

\begin{rem}
The theorem also holds when the notion of PCSP is extended
to when the $P_i$'s and $Q_i$'s are related by a homomorphism instead
of inclusion.
\end{rem}

\begin{rem}
For CSPs that it is known that any relation $\Gamma$ there
is a digraph $H_{\Gamma}$ such that $\CSP(\Gamma)$ is {\emph{logspace}-equivalent} to $\Hom(H_{\Gamma})$ \cite{bulin2015finer}. The authors
conjecture that this is also the case for PCSPs.
\end{rem}

\subsection{{Lack of Repetition Does Not Make Things Harder}} \label{app:repetition}

For a set $\Gamma$ of promise relations, let $\PCSP_R(\Gamma)$ be the promise decision problem analogous to $\PCSP(\Gamma)$ except that each clause has at most one copy of each variable. {We} show that the two problems are polynomial-time equivalent using a simple combinatorial argument, simplifying the argument used in \cite{DBLP:journals/siamcomp/AustrinGH17} for establishing the $\mathsf{NP}$-hardness of ``balanced $2$-coloring'' versus ``weak $2$-coloring'' of $2k+1$-uniform hypergraphs.

\begin{thm}
For all finite $\Gamma = \{(P_i, Q_i) \in D^{k_i} \times D^{k_i}\}$, $\PCSP_R(\Gamma)$ is polynomial-time equivalent to $\PCSP(\Gamma)$
\end{thm}
\begin{proof}
  $\PCSP_R(\Gamma)$ trivially reduces to $\PCSP(\Gamma)$ since any instance of $\PCSP_R(\Gamma)$ is an instance of $\PCSP(\Gamma)$. Thus, we now consider the harder case. Let $\Psi = (\Psi_P, \Psi_Q)$ be a $\Gamma-\PCSP$ with $m$ clauses on the variable set $x_1, \hdots, x_n$. Let $k$ be the maximum arity of any promise relation of $\Gamma$ (this is a constant). For our reduction, replace each variable $x_i$ with $|D|k$ `copies' $x_i^{(1)}, \hdots, x_i^{(|D|k)}$. Replace each clause $P_i(x_{j_1}, \hdots, x_{j_{k_i}})$ of $\Psi_P$ with a conjunction of at most $(|D|k)^{k_i}$ clauses, $P_i(x_{j_1}^{(a_{1})}, \hdots, x_{j_{k_i}}^{(a_{k_i})})$ in which we remove the clauses with a repeated variable. Call this new formula $\Psi_P^R$. Perform an identical reduction of $\Psi_Q$ to $\Psi^R_Q$. We can see that $\Psi^R = (\Psi_P^R, \Psi_Q^R)$ is a valid $\Gamma$-PCSP without repetition and the size of this PCSP is only a constant factor larger than the size of $\Psi$.

  Now we show that this is a valid reduction. First, if $\Psi$ is satisfiable, there is an assignment to the variables $x_1, \hdots, x_n$ which satisfies $\Psi_P$. If we let each copy $x_i^{(j)}$ have the same value of $x_i$, then we yield a satisfying assignment of $\Psi_P^R$. It suffices then to show that if $\Psi$ is unsatisfiable, then $\Psi^R$ is unsatisfiable. This is equivalent to showing that if $\Psi_Q^R$ is satisfiable, then $\Psi_Q$ is also satisfiable. Assume we have a satisfying assignment of $\Psi_Q^R$. For each of the variables $x_i$ of $\Psi_Q$, set $x_i$ to be the most frequently occurring value in the multiset $\{x_i^{(j)} : j \in \{1, \hdots, |D|k\}\}$ (break ties arbitrarily). Crucially note that this most frequently occurring value occurs at least $k$ times. We claim that this choice of the $x_i$ satisfies $\Psi_Q$. For each clause $Q_i(x_{j_1}, \hdots, x_{j_{k_i}})$, we can find a corresponding clause $Q_i(x_{j_1}^{(a_1)}, \hdots, x_{j_{k_i}}^{(a_{k_i})})$ in $\Psi_Q^R$such that $x_{j_{\ell}}^{a_{\ell}} = x_j$ for all $j$ (this is possible without repetition since there are at least $k$ distinct choices for $a_{\ell}$). Since $\Psi_Q^R$ is satisfied, this particular repetition-free clause is satisfied, so the corresponding clause in $\Psi_Q$ is satisfied. Thus, we have found a satisfying assignment for $\Psi_Q$. Thus, $\PCSP_R(\Gamma)$ and $\PCSP(\Gamma)$ are polynomial-time equivalent. 
\end{proof}

\section*{Acknowledgments}

The authors would like to thank anonymous referees for useful
comments, including simplifying/correcting a few of the proofs. The authors also thank Andrei Krokhin and Victor Dalmau for
important comments.

\bibliographystyle{plain}
\bibliography{pcsp2}

\end{document}